\documentclass[letterpaper,twocolumn,10pt]{article}
\usepackage{usenix-2020-09}

\usepackage{tikz}
\usepackage{amsmath}

\usepackage{adjustbox}
\usepackage{amsthm}
\usepackage{caption}
\usepackage{changepage}
\usepackage{datetime}
\usepackage{enumitem}
\usepackage{graphicx}
\usepackage{framed}
\usepackage{hyperref}
\usepackage{multirow}
\usepackage{url}
\usepackage{breakurl}

\usepackage{verbatim} 
\usepackage{xspace} 
\usepackage[normalem]{ulem}
\usepackage{booktabs,array,ragged2e}

\usepackage{color, colortbl}
\definecolor{grey}{gray}{0.9}
\definecolor{lgrey}{gray}{1}
\definecolor{white}{rgb}{1,1,1}
\definecolor{dgreen}{HTML}{228B22}

\newlength\mylength

\usepackage[font=large, position=bottom]{subfig}
\usepackage{comment}

\hypersetup{
	colorlinks=true,
	linkcolor=black,
	filecolor=blue,
	citecolor = red,      
	urlcolor=blue,
}

\usepackage{microtype}
\usepackage{graphicx}
\usepackage{booktabs} 

\usepackage[utf8]{inputenc}
\usepackage{amsthm}
\usepackage{amsmath}
\usepackage{algorithmic}
\usepackage{titling}

\theoremstyle{plain}
\theoremstyle{definition}

\newtheorem{theorem}{Theorem}[section]

\newcommand{\mycaption}[2]{\caption{\textbf{#1}. {#2}}}
\newcommand{\sref}[1]{\S\ref{#1}}
\newcommand{\vheading}[1]{\vspace{0.05in}\noindent\textbf{#1}}
\newcommand{\myx}{$\times$\xspace}

\newcommand{\minio}{MinIO\xspace}

\newcommand{\sysname}{Synergy\xspace}
\newcommand{\jsplit}{\textit{split}\xspace}
\newcommand{\fair}{GPU-proportional\xspace}
\newcommand{\systune}{Synergy-\textsc{Tune}\xspace}
\newcommand{\sysopt}{Synergy-\textsc{Opt}\xspace}
\newcommand{\sysgreedy}{Synergy-\textsc{Greedy}\xspace}

\newcommand{\ie}{\textit{i.e.,}\xspace}
\newcommand{\eg}{\textit{e.g.,}\xspace}

\newcommand{\upto}{up to\xspace}

\newcommand{\begincompactitemize}{\begin{itemize}[noitemsep,topsep=0pt,parsep=0pt,partopsep=0pt]}

\newcommand{\revised}[1]{\textcolor{black}{#1}}

\usepackage{ifthen}
\newboolean{publicversion}
\newboolean{submissionversion}

\setboolean{publicversion}{false}
\setboolean{submissionversion}{false}
\IfFileExists{intro-flag}
{
	\setboolean{publicversion}{True}
}
{}

\ifthenelse{\boolean{publicversion}}{
	\newcommand{\grumbler}[3]{}
}{
	\newcommand{\grumbler}[3]{\xspace\textcolor{#3}{\bf #1: #2}}
}

\newcommand{\ra}[1]{\renewcommand{\arraystretch}{#1}}

\newcommand{\jm}[1]{\grumbler{Jayashree}{#1}{magenta}}

\newcommand{\todo}[1]{\xspace\textcolor{red}{\bf #1}}

\usepackage[absolute]{textpos}

\newcommand{\ut}{{\large$^\dag$}}
\newcommand{\msr}{{\large$^\star$}}
\newcommand{\vmware}{{\large$^\ddag$}}

\begin{document}
\date{}

\title{\Large \bf Synergy: Resource Sensitive DNN Scheduling in Multi-Tenant Clusters}

\author{
	Jayashree Mohan\msr\thanks{\footnotesize{Work done as a MSR intern in Project Fiddle.}}\hspace{0.07in}, Amar Phanishayee\msr, Janardhan Kulkarni\msr, Vijay Chidambaram\ut\vmware\\
	\rm{\textit{\msr Microsoft Research\hspace{0.04in} \ut University of Texas at Austin\hspace{0.04in} \vmware VMware Research}}
}

\include{stdinc}

	\maketitle
 Training  Deep Neural Networks (DNNs) is a popular workload in both enterprises and cloud data centers. Existing schedulers for DNN training consider GPU as the dominant resource and allocate other resources such as CPU and memory proportional to the number of GPUs requested by the job. Unfortunately, these schedulers do not consider the impact of a job's sensitivity to allocation of CPU and memory resources. In this work, we propose \sysname, a 
 resource-sensitive scheduler for shared GPU clusters. \sysname infers the sensitivity of DNNs to different resources using optimistic profiling; some jobs might benefit from more than the \fair allocation and some jobs might not be affected by less than \fair allocation. \sysname performs such multi-resource workload-aware assignments across a set of jobs scheduled on shared multi-tenant clusters using a new near-optimal online algorithm.  Our experiments show that workload-aware CPU and memory allocations can improve average job completion time by upto 3.4\myx, by better utilizing existing cluster resources, compared to traditional \fair scheduling.


	\section{Introduction}
\label{sec-intro}

The widespread popularity of Deep Neural Networks (DNNs) makes training such models an important workload in both enterprises and cloud data centers.
Training a DNN is resource-intensive and time-consuming. Enterprises typically setup large multi-tenant clusters, with expensive hardware accelerators like GPUs, to be shared by several users and production groups~\cite{jeon2018multi, antman}. 
In addition to the model-specific parameters and scripts, jobs specify their GPU demand before being scheduled to run on available servers. Jobs are scheduled and managed either using traditional big-data schedulers, such as Kubernetes~\cite{kubernetes} or YARN~\cite{yarn}, or using modern schedulers that exploit DNN job characteristics for better performance and utilization~\cite{gavel, gandiva, themis, tiresias, allox, optimus, gandivafair}. These DNN schedulers decide how to allocate GPU resources to many jobs while implementing complex cluster-wide scheduling policies to optimize for objectives such as average job completion times (JCT), makespan, or user-level fairness. 

\begin{figure}[!htb]
	\centering
	\includegraphics[width=.4\textwidth]{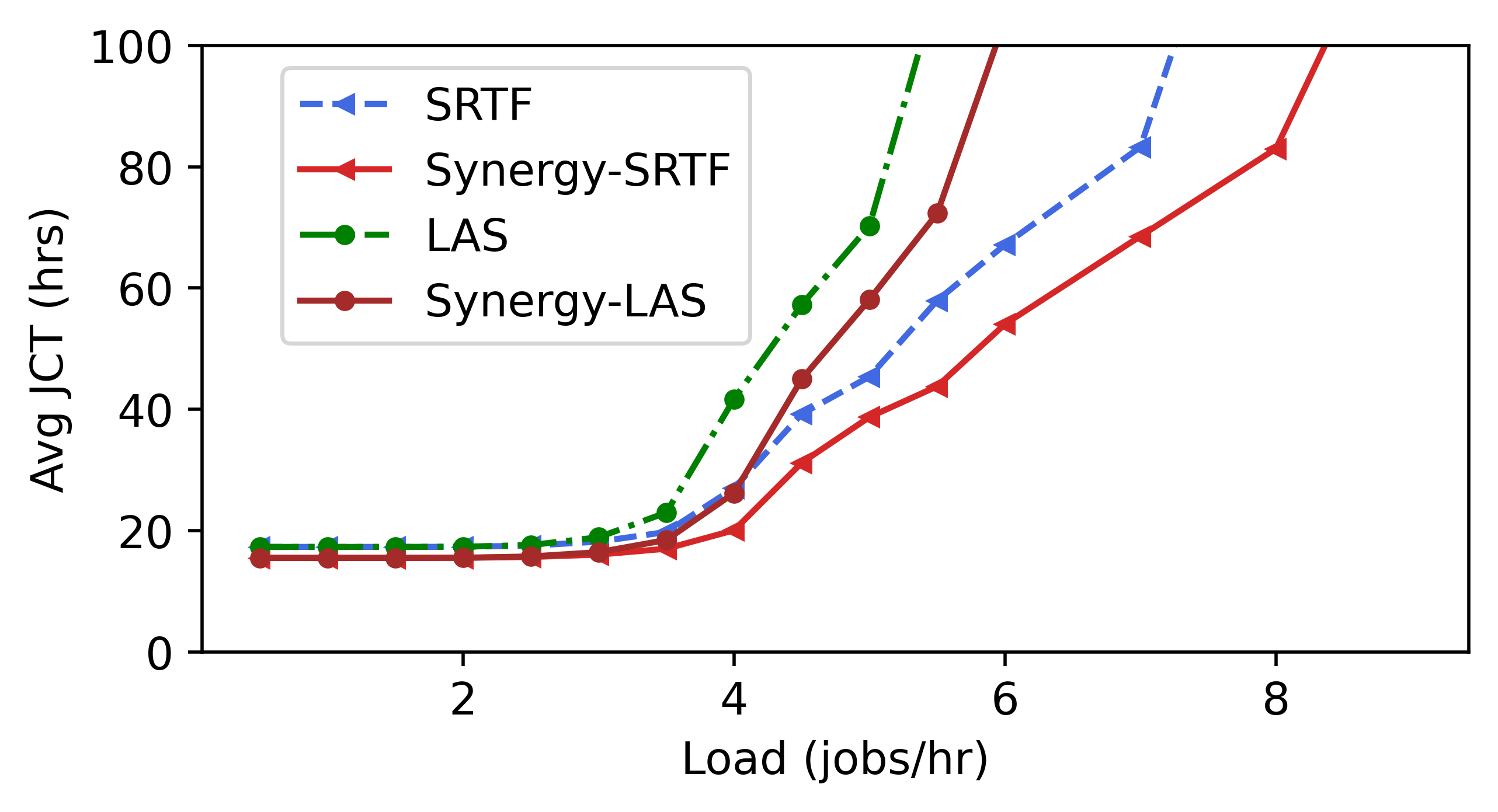}%
	\vspace{-1em}
	\mycaption{Average JCT with \sysname}{\sysname is able to significantly reduce average JCT and support higher load for different scheduling policies (shown here on a cluster of 128 GPUs for a Philly-derived trace as we vary load~\cite{philly}). }
	\label{fig-mot}
	\vspace{-1em}
\end{figure}
Current DNN cluster schedulers assume GPUs to be the dominant resource in the scheduling task~\cite{jeon2018multi, gavel, gandiva, themis, tiresias, allox, optimus, gandivafair}; i.e., a user requests a fixed number of GPUs for her DNN job, and when the requested number of GPUs are all available, the job is scheduled to run. Other resources such as CPU and memory are allocated proportional to the number of GPUs assigned to the job (\textit{\fair} allocation). 

However, we identify an important property of DNN training jobs that \fair allocation is unable to exploit: DNNs exhibit varied sensitivity to the amount of auxiliary resources like CPU and  memory allocated to the job. Prior work has shown that ingesting data for ML training jobs, i.e., reading data from storage to memory, and pre-processing them at the CPU is computationally expensive, thereby resulting in \emph{data stalls} in both research~\cite{coordl} and industry scale training at large enterprises such as Google~\cite{tfdata} and Facebook~\cite{fbstalls}. For instance, some image and video recognition models achieve \upto 3\myx speedup by overcoming \emph{data stalls} (\sref{sec-motivation}) when the CPUs allocated exceed their \fair share, while other models like GNMT are unaffected when the CPUs assigned are less than \fair share. 

Our main insight here is that allocating these auxiliary resources in a workload-aware fashion, rather than the traditional \fair allocation can significantly improve performance by effectively utilizing \textit{cluster-wide} resources. Based on this insight, we propose \emph{\sysname}, a resource-sensitive scheduler for \revised{homogeneous}, multi-tenant GPU clusters. Figure~\ref{fig-mot} shows the average job completion time (JCT) in the cluster as we vary load, for two scheduling policies; \sysname's resource-sensitive allocation is able to significantly improve average JCT in the cluster and sustain a higher load compared to \fair allocation.



 \sysname profiles the sensitivity of DNNs to  auxiliary resources and allocates them disproportionately among jobs rather than using traditional \fair allocation. While doing so, \sysname ensures that a job gets less than \fair auxiliary resources \textit{only} if such an allocation does not degrade the job throughput compared to a \fair allocation. Such allocation enables \sysname to mitigate data stalls in several models, thereby significantly increasing the overall cluster throughput.

Efficiently exploiting the heterogeneity in resource sensitivity among DNN jobs raises two important problems which have not been tackled by prior work:
\begincompactitemize
\item What is the ideal resource requirement for each job (with fixed GPU demand) and how can this be determined with low overhead? 
\item How should we pack these jobs onto servers along multiple resource dimensions efficiently, especially when we can tune the job's demand for these resources?
\end{itemize}

\vheading{Optimistic profiling}. 
\sysname exploits the predictability of DNN computation to measure the job throughput as we vary the amount of CPU and memory allocated to the job. This is performed offline by the \sysname scheduler, prior to job execution on the cluster. However, profiling all possible combinations of CPU, and memory values is computationally expensive. Therefore, \sysname introduces optimistic profiling; it empirically profiles the job throughput for varying CPU allocations, assuming maximum memory allocation. It then \textit{analytically} estimates the job throughput for all combinations of CPU and memory.
A key insight that makes such analytical modelling feasible is the predictable nature of job performance to memory allocation when using DNN-aware caching like \minio~\cite{coordl} that guarantees a certain cache hit rate. We show in ~\sref{sec-profile} that our optimistically profiled model performance closely resembles the true empirical values, while significantly reducing profiling time (by up to 30\myx). Using these profiles, \sysname identifies the best resource allocation 
beyond which the job throughput has diminishing returns.

\vheading{Scheduling mechanism}. \sysname makes a round-based scheduling decision similar to prior DNN schedulers~\cite{gavel}. 
In each round (say 5 minutes), we identify the set of jobs that are runnable in the cluster using a scheduling policy such as FIFO~\cite{spark, yarn}, SRTF~\cite{srtf}, LAS~\cite{las, tiresias}, FTF~\cite{themis}, etc. 
\sysname's scheduling mechanism then packs these jobs among available servers in the cluster along all resource dimensions identified in the profiling phase.  This is analogous to multi-dimensional bin-packing problem, which is NP-Hard~\cite{np_hard}, and hence requires approximate solutions. But unlike prior work in big-data scheduling which tackles the problem of multi-dimensional bin-packing with fixed resource demands (for \eg Tetris~\cite{tetris}, DRF~\cite{drf}), \sysname has to contend with fungible resource demands.
This introduces two challenges that need to be solved in tandem: First to find an optimal partition of CPU and memory among jobs to maximize throughput while ensuring fair allocations (every job's throughput is at least that of \fair allocation), and second, a feasible packing of these resources among jobs.


In this paper, we propose two effective algorithms to enable such fungible multi-dimensional bin-packing.
Our first algorithm, \sysopt, is formulated as a linear program and enables determining an upper-bound on achievable throughput by an optimal solution for a given workload trace.
However, we find that \sysopt is impractical for two reasons: (1) it is computationally expensive as we scale cluster size, 
and (2) it produces fractional GPU allocations that cannot be achieved in real deployments. Nevertheless, its solution provides an aspirational optimal goal that we can use to measure the efficacy of any practical solution.
The second algorithm, \systune, is fast and near-optimal (within 10\% of \sysopt in evaluation). If a job to be scheduled does not fit in the cluster along all the resource dimensions, we revert the job demands to \fair if its current demands are above it. If the job's demands are already \fair or below, then we find a suitable job in the cluster with higher than \fair allocation, which is then reverted to \fair.  \systune also outperforms simpler greedy approaches (\sysgreedy) that recursively pack jobs along multiple resource dimensions using a first-fit allocation strategy~\cite{first_fit}.

We implement a prototype of \sysname and an accompanying event-driven simulator in Python. \sysname transparently communicates with the DNN job using a thin iterator API, that is a wrapper around the existing data iterator, thereby requiring minimal code changes to the DNN job script. Across various scheduling policies, and workload traces, we show that \sysname improves cluster objectives such as average JCT by \upto 1.5\myx on a physical cluster of 32 GPUs. On a large simulated cluster of up to 512 GPUs, \sysname improves average JCT by \upto 3.4\myx. \sysname is open sourced at \href{https://github.com/msr-fiddle/synergy}{https://github.com/msr-fiddle/synergy}.

\begin{figure*}[!tbh]

  \subfloat[CPU sensitivity\label{fig-cpu-sensitivity}]{{\includegraphics[width=0.7\textwidth]{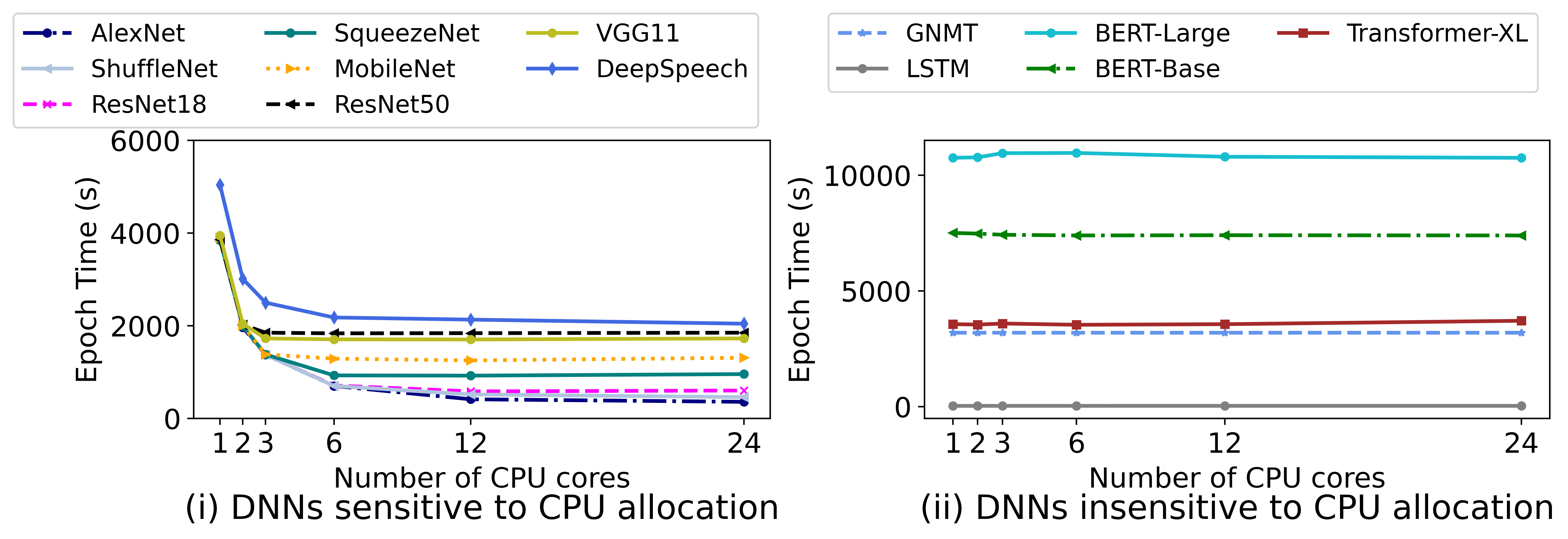} }}
  \quad
\subfloat[GPU VM SKUs\label{fig-mot-sku}]{
\ra{0.8}
    \begin{tabular}[b]{cc}
    \toprule
      \textbf{CPU:} & \multirow{ 2}{*}{\textbf{SKU}} \\
       \textbf{GPU} & \\
      \midrule
      \multirow{ 2}{*}{3:1} & NVIDIA DGX-2\\
         & Internal servers at X\\
         \midrule
      4:1 & AWS p3.16xlarge \\
      \midrule
      \multirow{ 2}{*}{5:1} & NVIDIA DGX-1 \\
            & Azure NDv2 \\
      \midrule
      6:1 & Azure NC24s\_v3 \\
    
      \bottomrule
    \end{tabular}
}

    \mycaption{CPU sensitivity}{This graph plots the epoch time for DNNs as we vary the CPU:GPU ratio for single-GPU training.  Some jobs such as Transformers need as few as 1 CPU core per GPU to achieve maximum training speed; others like ShuffleNet need more than 12 CPU cores per GPU to eliminate data stalls. State-of-the-art GPU VMs have a CPU:GPU ratio as few as 3.}

\label{fig-mot-cpu}
\end{figure*}

In summary, our paper makes the following contributions.
\begin{itemize}[leftmargin=*,,noitemsep,partopsep=0pt,topsep=0.2em,parsep=0pt]
\item We identify the importance and need for resource-sensitive scheduling of DNN jobs in multi-tenant GPU clusters (\sref{sec-motivation}).
\item We present \sysname, a resource-sensitivity aware scheduler that optimistically profiles the job's resource demands and performs disproportionate allocations such that no job achieves lower than \fair throughput (\sref{sec-design}).
\item We present a heuristic scheduling mechanism \systune, that maps the allocations calculated by the profiler onto the cluster, while better utilizing the resources compared to a \fair allocation (\sref{sec-sched-mechanism}).
\item In extensive experimentation on physical and simulated clusters, \sysname's techniques improve average JCT by \upto 3.4\myx, thus supporting a higher input load (\sref{sec-eval}).
\end{itemize}

	
\section{Background and Motivation}
\label{sec-motivation}
In this section, we briefly describe DNN scheduling, introduce the terminology used in the rest of the paper, and motivate resource-sensitive DNN cluster scheduling.


\vheading{Scheduling ML training jobs in a cluster}. Training a ML model is a resource intensive and long-running task (order of hours to days). Collocating ML training workloads in a shared, multi-tenant cluster is a very common setup in several large organizations, for both research and production~\cite{gandiva, themis, tiresias, optimus, gavel}. Our work targets state-of-the-art multi-tenant clusters similar to the ones published by prior large-scale studies by organizations like Microsoft~\cite{jeon2018multi} and Alibaba~\cite{antman}. These clusters use on-premise servers or cloud VMs with pre-defined GPU, CPU, and memory resources.  The cluster itself is shared by multiple users and jobs, and each server can host more than one job each with varying resource usage (some heavy on CPU side pre-processing, while others heavy on GPU computation). For example, a server with 8 GPUs can host 8 single-GPU jobs from different users.

\vheading{Scheduling policy and mechanism}.
When jobs are submitted to a scheduler, a scheduling policy such as First In, First Out (FIFO)~\cite{spark, yarn}, Shortest Remaining Time First (SRTF)~\cite{srtf}, Least Attained Service (LAS)~\cite{las, tiresias}, or Finish Time Fairness (FTF)~\cite{themis} decides the set of jobs ($J$) to be run on the cluster. A scheduling mechanism then identifies where job $J$ should be run, and how much resources to allocate to the job. The GPU demand for a job is fixed (requested by the user), while the CPU and memory allocation is fungible. 

\vheading{\fair allocation.}
During DNN training, a minibatch of data is first fetched from storage to memory, where it is cached for subsequent accesses. It is then pre-processed at the CPU, and then copied over to the GPU for processing. Existing DNN schedulers~\cite{gandiva, gavel, themis, tiresias}, and those used in real-world GPU clusters~\cite{jeon2018multi, philly}, including recent schedulers that offer GPU elasticity~\cite{afs, pollux}, all allocate CPU and memory resources to a job using a \fair allocation. For instance, consider a server with 4 GPUs, 16 CPUs and 200 GB memory. If a job requests 1 GPU, then it is allocated 4 CPUs and 50GB memory. 

\begin{table*}[t!]
\ra{1}
\begin{minipage}[b]{.20\textwidth}
	\centering
	\begin{tabular}{ | l | l |}
		\hline
		Job & Model \\ \hline
		$J_1$ & ResNet18 \\ \hline
		$J_2$ & Audio-M5 \\ \hline
		$J_3$ & Transformer \\ \hline
		$J_4$ & GNMT \\ \hline
	\end{tabular}
	\captionof{table}{Example jobs}
	\label{tbl-motivation-ex-jobs}
\end{minipage}\qquad
\begin{minipage}[b]{.35\textwidth}
	\centering
	\begin{tabular}{ |c | c | c | c | c | }
		\hline
		Server & Job & GPU & CPU & Mem \\ \hline
		\multirow{2}*{$S_1$} & $J1$ & 4 & 12 & 250 \\ \cline{2-5}
		& $J2$ & 4 & 12 & 250 \\ \hline
		\hline
		\multirow{2}*{$S_2$} & $J3$ & 4 & 12 & 250 \\ \cline{2-5}
		&$J4$ & 4 & 12 & 250 \\ \hline
	\end{tabular} \\
	\captionof{table}{\fair allocation}
	\label{tbl-motivation-ex-sched-1}
\end{minipage}\qquad
\begin{minipage}[b]{.35\textwidth}
	\centering
	\begin{tabular}{ |c | c | c | c | c | }
		\hline
		Server & Job & GPU & CPU & Mem \\ \hline
		\multirow{2}*{$S_1$} & $J1$ & 4 & \textcolor{blue}{23} & \textcolor{blue}{400} \\ \cline{2-5}
		& $J3$ & 4 & \textcolor{red}{1} & \textcolor{red}{100} \\ \hline
		\hline
		\multirow{2}*{$S_2$} & $J2$ & 4 & 12 & \textcolor{blue}{450} \\ \cline{2-5}
		&$J4$ & 4 & 12 & \textcolor{red}{50} \\ \hline
	\end{tabular} \\
	\captionof{table}{Resource-sensitive allocation}
	\label{tbl-motivation-ex-sched-2}
\end{minipage}\qquad
\vspace*{-3em}
\end{table*}

\subsection{Motivation : Resource sensitivity}
\label{sec-res-sensitivity}

\vheading{Insight}. The main insight that motivates our work is that DNNs co-scheduled on a cluster exhibit different levels of sensitivity to CPU and memory allocations during training. 
Therefore, it is possible to improve the overall cluster utilization and efficiency by performing resource-sensitive allocations instead of the ubiquitously used \fair allocation. Prior work on characterization study of jobs in Microsoft’s Philly cluster~\cite{jeon2018multi} shows that CPU cycles are  under-utilized in multi-tenant clusters; we use this as motivation to show that we can \emph{exploit the disparity in resource requirements \textbf{across jobs} to improve overall cluster utilization without any hardware upgrades (storage, CPU, or memory)}.

Figure~\ref{fig-cpu-sensitivity} plots the per-epoch time for various DNNs when trained on a single GPU 
by varying the number of CPUs allocated to the job (ensuring that the dataset is fully cached for each job). 
Figure~\ref{fig-cpu-sensitivity}(i) shows that most image and speech models are sensitive to CPU allocations; smaller models like ShuffleNet and ResNet18 require 9--24 CPU cores per GPU to pre-process data items. However, state-of-the-art ML optimized servers and cloud GPU VMs have a CPU:GPU ratio as few as 3 as shown in Table~\ref{fig-mot-sku}~\cite{dgx2,aws-p3,dgx1,az-ndv2,az-ncv3,  refurbish}. Increasing the CPU:GPU ratio from 3 to 12 results in 3.1\myx faster training for AlexNet, and increasing it to 9 results in 2.3\myx faster training for ResNet18. On the other hand, most language models are insensitive to CPU allocations as shown in Figure~\ref{fig-cpu-sensitivity}(ii). This is because they have modest input data pre-processing requirements. Transformer models for example,  unlike image classification models, do not perform several unique data augmentation operations for each data item in every epoch~\cite{coordl}. 

Next, to understand the importance of memory allocations, we train two models; an image classification model - ResNet18 on  OpenImages~\cite{openimages} and a language model GNMT on WMT, with varying memory allocations on a server whose \fair share of memory per GPU is 62GB. 
We observe that GNMT is insensitive to memory allocation; even if only 20GB memory is allocated (which is the required process memory for training), the training throughput is unaffected. However, increasing the memory from 62GB (\fair allocation)  to 500GB (max) for ResNet18 speeds up training by almost 2\myx. This is because, language models like GNMT, and transformers are GPU compute bound. Therefore, fetching data items from storage if they are not available in memory does not affect training throughput. On the other hand, image and speech models benefit from larger DRAM caches. If a data item is not cached, the cost of fetching it from the storage device can introduce fetch stalls in training~\cite{coordl, tfdata, fbstalls}. 

\vheading{Takeaway}. \textit{When two jobs have to be scheduled on the same server, it is possible to co-locate a CPU-sensitive job with a CPU-insensitive one. This allows CPU allocation to be performed in a resource-sensitive manner rather than  \fair allocation. Similarly, it is always beneficial to pack a memory-sensitive job with an insensitive one, allowing disproportionate resource-sensitive sharing of memory to improve the aggregate cluster throughput. }

 \begin{figure}[!t]

  \centering 

  \includegraphics[width=0.95\columnwidth]{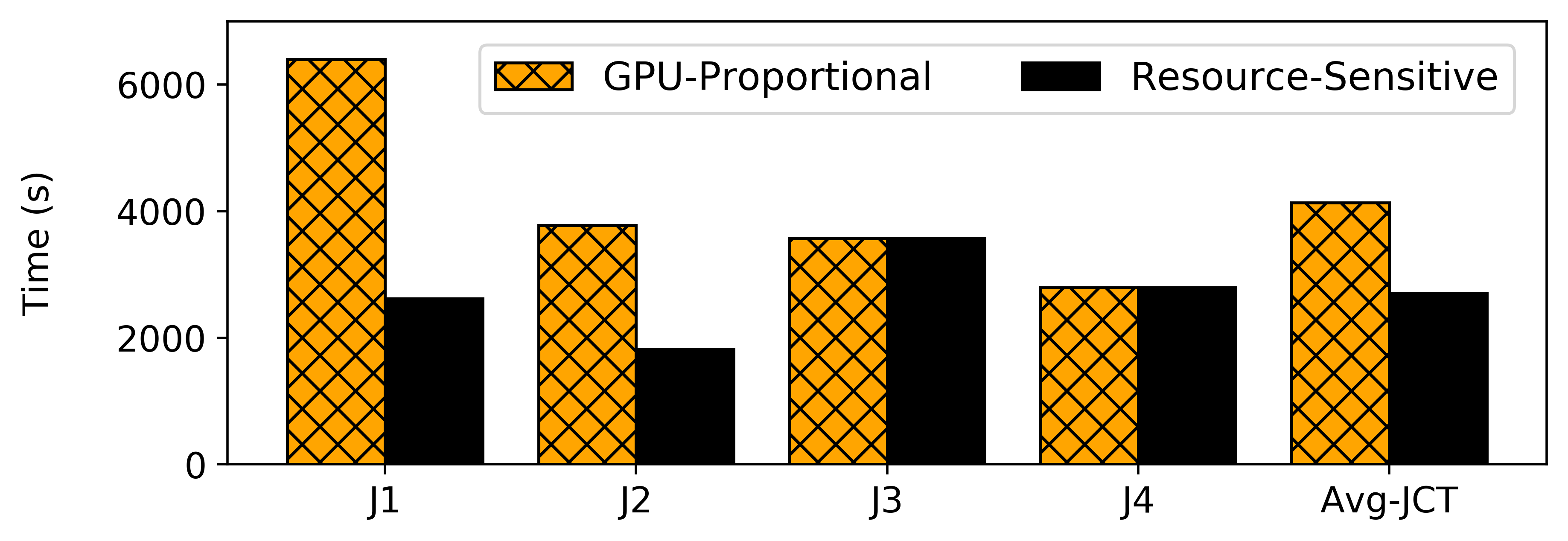} 
   
    \vspace{-1em}
    \mycaption{Resource sensitive scheduling}{We compare the runtime of the jobs with two different schedules; \fair and resource-sensitive. By allocating resources disproportionately, CPU and memory sensitive jobs see increased throughputs which reduces the average JCT by 1.5\myx.}

  \label{fig-motivation-ex}
  \vspace{-1em}
\end{figure}
\vheading{Example}. We now show how resource-sensitivity-aware scheduling can improve cluster efficiency using a simple example. We run the experiment on two physical servers each with 8 GPUs, 24 CPUs and 500GB DRAM (internal servers at a large cloud provider X). Let's say we have 4 jobs in the scheduling queue, each requesting 4 GPUs as shown in Table~\ref{tbl-motivation-ex-jobs}. 
We consider two different schedules; (1) \fair allocation and (2) resource-sensitive allocation. The results of these schedules are shown in Table~\ref{tbl-motivation-ex-sched-1} and Table~\ref{tbl-motivation-ex-sched-2}. Figure~\ref{fig-motivation-ex} compares the epoch time of each of these jobs in the two scenarios. The increased resource allocation to CPU and memory sensitive jobs in Schedule 2 
speeds up $J1$ and  $J2$ significantly, while leaving the runtime of $J3$ and $J4$ unaffected. The average JCT in the cluster thus drops by 1.5\myx due to resource-sensitive allocations.



\subsection{\sysname Scheduling Policies}
\sysname is not constrained to one particular scheduling policy, but is instead general enough to improve a wide range of scheduling policies (\eg LAS, FIFO, SRTF, FTF, etc), creating \sysname-augmented variants for all of them. The main challenge that \sysname addresses is, finding an efficient partition of available cluster CPU and memory among jobs to maximize throughput while ensuring that every job’s throughput is at least that of \fair allocation. \sysname’s innovation thus lies in exploiting the differences in resource sensitivity across jobs to improve overall cluster metrics. 

\subsection{\revised{Assumptions \& Limitations}}

\revised{In the context of this work, we explicitly highlight certain practical assumptions, many of which are derived directly from large multi-tenant clusters we  analyze - homogeneous clusters, fixed GPU allocation for the lifetime of a job, and the use of MinIO cache. \sysname’s design is not tied to these assumptions, but it aids in focused profiling (reducing the dimensionality of the search space). In a large scale, multi-tenant, production cluster, it is practical to assume that there are tens of thousands of accelerators per homogeneous cluster, and the GPU allocation for a job remains constant. While recent works explore scheduling DNN jobs in heterogeneous clusters~\cite{gavel, gandivafair, allox}, and GPU elasticity~\cite{pollux}, there are several practical challenges in seamlessly supporting these features. For instance, with elastic training, the impact of changing batch sizes and  hyperparameters on training accuracy is unclear for a wide variety of tasks. We provide a detailed discussion on the practicality of each of these assumptions made by \sysname, and what it means to relax these assumptions for \sysname
in Section ~\ref{sec-disc}.}

	\section{\sysname : Design}
\label{sec-design}
\vheading{Overview}. \sysname is a round-based scheduler that arbitrates multi-dimensional resources (GPU, CPU, and memory) in a \textit{homogeneous} cluster. 
\sysname augments existing scheduling policies with \emph{resource sensitivity}  in two steps. First, it identifies the job's best-case CPU and memory requirements using \textit{optimistic profiling} (\sref{sec-profile}). \sysname then identifies a set of runnable jobs for the given round using a scheduling policy (\eg SRTF, FTF, LAS, etc) such that their collective GPU demand is less than or equal to the GPUs available in the cluster. Then, using the profiled resource demands, \sysname packs these jobs on to the available servers along multiple resource dimensions using a near-optimal heuristic algorithm (~\sref{sec-sched-mechanism}).  At the end of a round, the set of runnable jobs are updated using the scheduling policy, and their placement decisions are recomputed. We now discuss both the components of \sysname in detail. Note that \sysname only alters the auxiliary resource allocations; GPU demands are left unaltered for the lifetime of a job and are provided as inputs by the user.

\subsection{Optimistic Profiling}
\label{sec-profile}
A DNN job is profiled for its resource sensitivity once per lifetime of the job, i.e. on job arrival. 
Each incoming job is profiled by varying the CPU and memory allocated to the job. A \emph{resource sensitivity matrix} is then constructed for discrete combinations of CPU and memory allocations as shown in Figure~\ref{fig-profiling}. Since DNN training has a highly predictable structure, empirically evaluating training throughput for a few iterations gives a fair estimate of the actual job throughput~\cite{gandiva, coordl}.

It is easy to see that naively profiling different combinations of CPU and memory can be very expensive. For instance, if the cost of profiling one combination of CPU, and memory for a job is 1 minute, then to profile all discrete combinations of CPU and memory (assuming allocation in units of 50GB) on a server with 24 CPUs and 500GB DRAM takes about 24*10 = 240 minutes (4 hours)!

 \begin{figure}[!t]

  \centering 

  \includegraphics[width=0.4\textwidth]{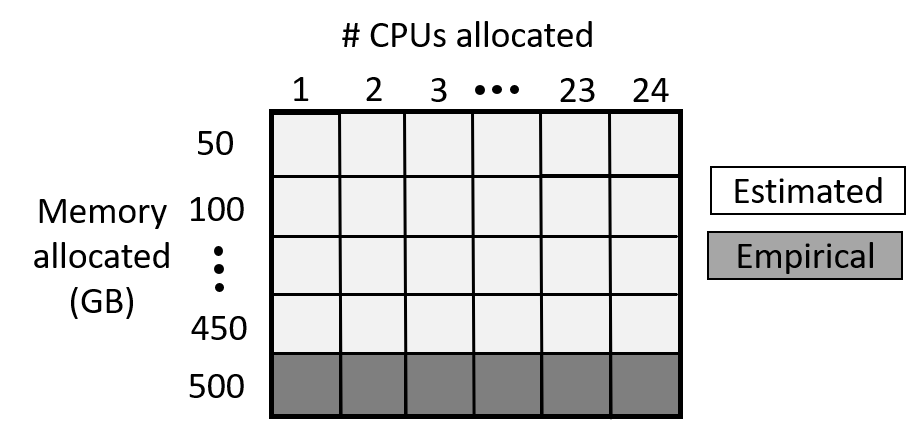}    
  \vspace{-1em}
    \caption{\textbf{Optimistic profiling} empirically evaluates the sensitivity of a model to
    varying \# CPUs assuming a fully cached dataset; the rest of the matrix is completed using estimation}
  \vspace{-3mm}
  \label{fig-profiling}

\end{figure}

To tackle this problem, \sysname introduces an optimistic profiling technique that exploits the predictability in the relationship between job throughput and memory allocation. We observe that, with DNN-specific, application-level caches like \minio~\cite{coordl}, it is easy to model the job throughput behaviour as we vary the amount of memory allocated to a job at fixed CPU allocation. This is because, \minio ensures that a job gets a fixed number of cache hits per epoch. \sysname makes a conscious decision to use application-level \minio cache instead of Page Cache because \minio provides memory isolation across independent jobs sharing the machine. If we do not use \minio, we will have to profile the model at discrete memory allocations which could result in increased profiling costs, and also potentially change the trends in profiling matrix.  However, the use of MinIO in \sysname makes cache performance predictable and hence reduces \sysname’s profiling costs -- allowing optimistic profiling.  

For a given CPU allocation that determines the pre-processing speed, and a known storage bandwidth, it is easy to analytically model the job throughput for varying memory allocation. Therefore, we only need to empirically profile the job for varying CPU values at full memory allocation as shown in Figure~\ref{fig-profiling}. 
All the other entries can be estimated using the above technique. This leads to a 10\myx reduction in profiling time, bringing it down to 24 minutes! We experimentally validate this in Figure~\ref{fig-validate-mem}. For a 8-GPU ResNet18 job, we compare the modeled job throughput using \sysname to the empirical results obtained by training the job for 2 epochs with varying memory allocations. As we see in  Figure~\ref{fig-validate-mem}, \sysname's estimations are within 3\% of the empirical results, without having to actually run the model.

To further optimize profiling time, we observe that we do not require exact throughput values for a job with varying CPU allocations. We instead need a curve depicting the empirical job throughput. Therefore, instead of profiling the job for all possible CPU values, we pick discrete points for CPU profiling using the following algorithm. We start with the maximum CPU allocation and do a binary search on the CPU values to estimate job throughput. If the profiled point resulted in a throughput improvement that is less than a fixed threshold (say 10\%), then we continue binary search on the lower half of CPU values, else we profile more points on the upper half. The idea here is to empirically profile CPU regions that show significant difference in job throughput, while skip those regions with little to no improvement in throughput. We experimentally show the efficacy of our CPU profiling technique in Fig~\ref{fig-validate-cpu} for a 1-GPU ResNet18 job. We compare the normalized job runtime (wrt 1 CPU) using empirical results averaged over 2 epochs of the job and \sysname's optimistic profiling averaged over  50 iterations (approximately, a minute per profile). \sysname is able to mimic the empirical job performance very closely, in under 8 minutes (using just 8 CPU profile points instead of 24). We believe that this is a reasonable overhead as it is incurred only once per lifetime of the job, which typically runs for hours.

\begin{figure}[!t]
	\centering

\subfloat[Memory Validation\label{fig-validate-mem}]{{\includegraphics[width=.23\textwidth]{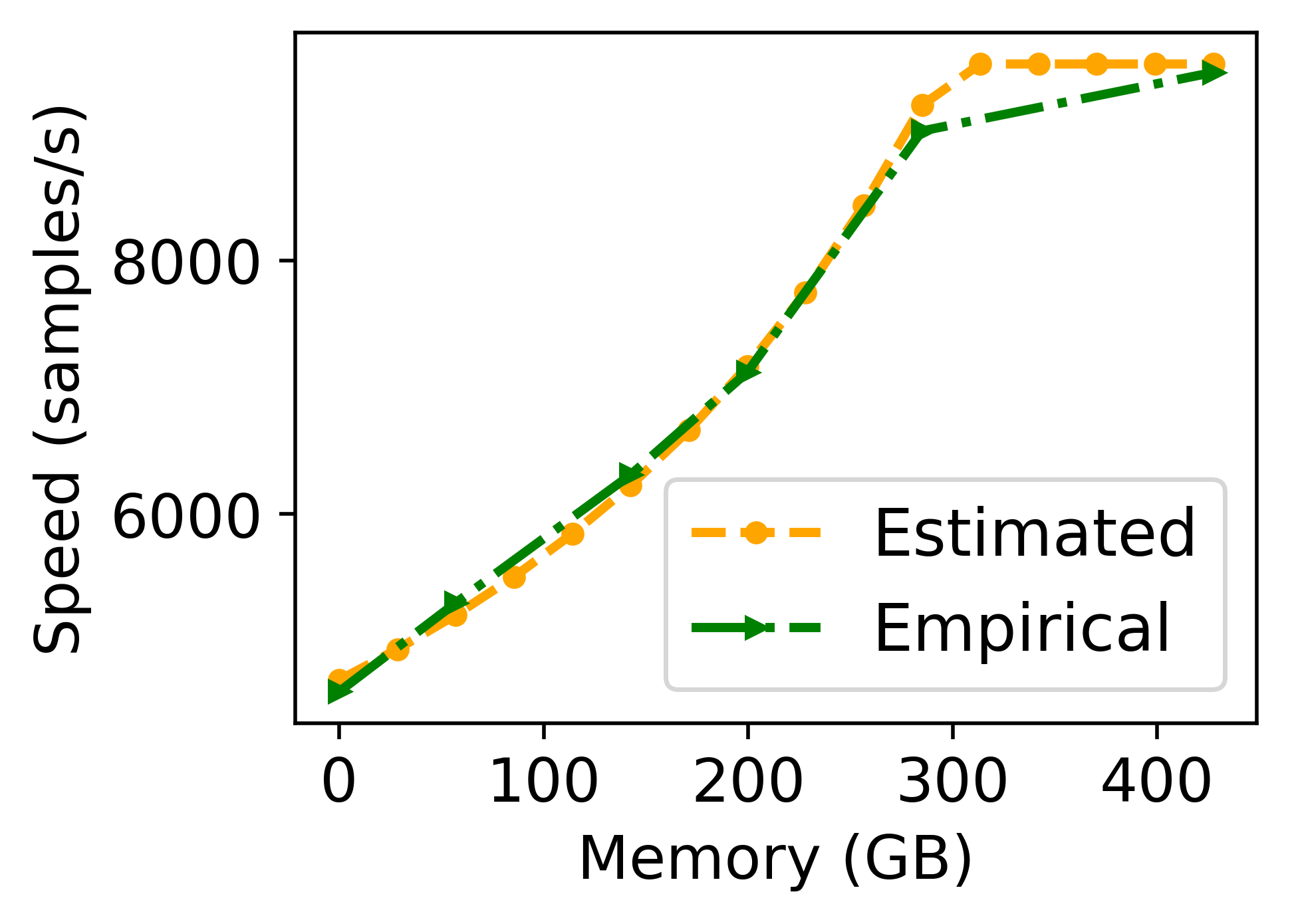} }}%
 \subfloat[CPU validation\label{fig-validate-cpu}]{{\includegraphics[width=.23\textwidth]{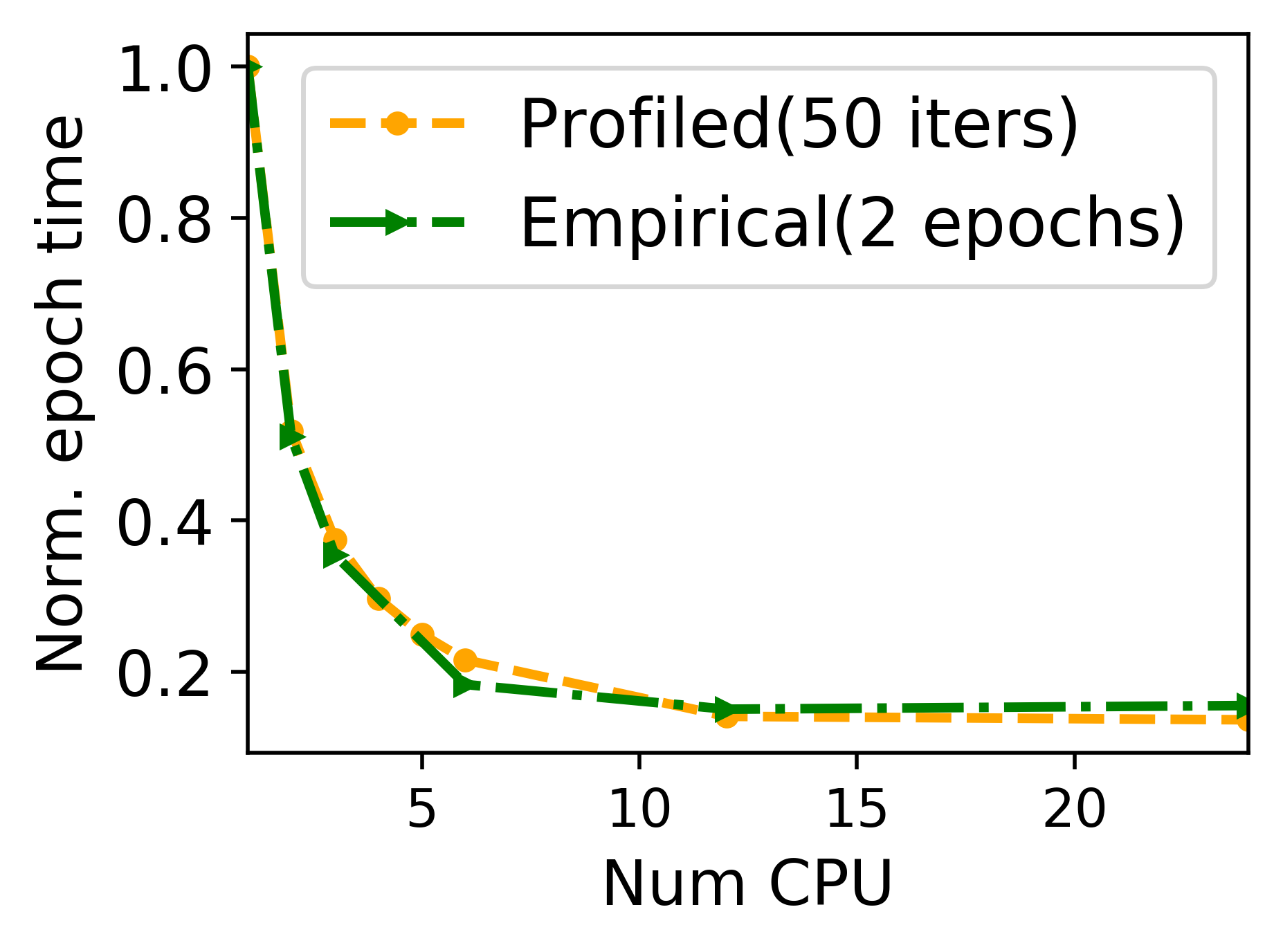} }}%
	
 \mycaption{Optimistic profiling}{The graphs compare the profiling results to empirical runs for ResNet18}
 \vspace{-3mm}
\label{fig-eval-prof-validate}
\end{figure}

After profiling a job on arrival, the job along with its resource sensitivity matrix is enqueued into the main scheduling queue, from which the scheduling policy picks a set of runnable jobs every round. 


\subsection{Scheduling mechanism}
\label{sec-schedule}
\sysname performs round-based scheduling. At the beginning of each scheduling round, 
\sysname identifies a set of runnable jobs from the scheduling queue that can be packed on the cluster in the current round duration using a scheduling policy such as FIFO, SRTF, LAS, or FTF. Using the resource sensitivity matrix, \sysname packs these jobs onto the available servers in the cluster while satisfying the multi-dimensional resource constraints as opposed to simply performing a \fair allocation . 

\vheading{Job demand vector}. To pack the jobs onto servers, we first construct a job demand vector that indicates the  GPU demand, and best-case CPU and memory requirements for the job. We identify the best-case values using the resource sensitivity matrix. We pick the minimum value of CPU and memory that saturates the job throughput. 

Packing a job with multi-dimensional resource demands is analogous to multi-dimensional bin packing problem which is NP hard~\cite{np_hard}. 
Therefore, we first evaluate the efficacy of a naive greedy scheduling mechanism as an approximation to tackle the multi-dimensional resource allocation problem.

\subsection{\sysgreedy : Greedy Scheduling}
A naive greedy multi-resource packing algorithm translates to a first-fit approximation of the multi-dimensional bin packing problem~\cite{first_fit}. Given a job demand vector, the greedy algorithm picks the next runnable job decided by the scheduling policy, and places it on the server that can satisfy the job's demands in all dimensions. If no such server exists, the job is \emph{skipped} over for this round and the next runnable job is checked for schedulability. \sysgreedy thus introduces two major problems in the cluster - 
\begincompactitemize
	\item It can result in auxiliary resources being exhausted by jobs, while leaving GPUs underutilized, and fragmented. 
	We show that GPU fragmentation in \sysgreedy severely degrades cluster objectives (\ref{sec-eval-split}). 
	
	\item It also hurts the fairness of the scheduling policy as some jobs can be skipped over for a long time if their resource demands cannot be satisfied in the cluster. 
\end{itemize}

 The challenge ahead of us is to design a scheduling mechanism that eliminates GPU under-utilization due to fragmentation, and upholds the fairness properties of the given scheduling policy, while performing multi-dimensional resource allocation. Before we come up with a  heuristic scheduling approach to tackle the above problems, one pertinent question is to understand how good is the allocation produced by our heuristic when compared to an optimal solution. 


To this end, we first formulate a theoretical upper bound on the optimal throughput achieved by the cluster given a set of jobs and their resource sensitivity profiles. We then discuss the challenges associated with materializing the optimal allocation on a physical cluster and introduce \systune, an empirically close-to-optimal heuristic solution.
	
\section{Scheduling Algorithms}
\label{sec-sched-mechanism}
We first present our formulation of an optimal allocation that provides an upper bound on the achievable cluster throughput. 

\subsection{\sysopt}
\label{sec-opt}
Our goal is to allocate CPU and memory to each job so as to maximize overall throughput, while guaranteeing that each job makes at least as much progress as it would do if we allocate its {\em \fair share}. It is not hard to show that our problem is NP-hard. So, we resort to finding approximate solutions using LP formulation. To find an \emph{upperbound} on achievable throughput, we solve two LPs. In the interest of space, we describe the first LP formulation here, and summarize the challenges in operationalizing \sysopt. A complete description of \sysopt formulation and proof can be found in the extended version~\cite{synergy-arxiv}. \revised{While the focus of this work is on homogeneous cluster, we show how our formulation can be extended to a heterogeneous GPU cluster in the extended version of the paper~\cite{synergy-arxiv}}.

\subsubsection{Finding ideal allocation}
First, we assume an {\em idealized setting}:  all the CPU and memory available across all the machines is present in one (super) machine. Say there are a total of $s$ homogeneous machines in the cluster. We assume that, there is only one machine with $G$ units of GPU, $C$ units of CPU, and $M$ units of memory. Note that, in reality $G_i$, $C_i$, and $M_i$ denote the total GPU, CPU, and memory in each machine $i$, which is $G/s$, $C/s$, and $M/s$ respectively in a homogeneous cluster. Based on this assumption, we find the ideal  CPU ($c^*_j$) and memory ($m^*_j$) allocation for every job $j$ (whose GPU demand is denoted by $g_j$) in the set of runnable jobs ($J_t$) for a round. 

The variables of our LP are denoted by $y_{\{c,m,j\}}$, which should be interpreted as follows.
If for a job $j \in J_t$,  $y_{\{c,m,j\}} = 1$, then it means that in the LP solution $c$ units of CPU and $m$ units of memory are allocated. We further note that for every job $j$, there is a variable $y_{\{c,m,j\}}$  for {\em for every possible} allocation of CPU and memory. We consider these variables in the discrete space as identified by our resource sensitivity matrix ($W_j$).  
$W_j[c, m]$ denotes the amount of progress made by job $j$ per round if $c$ units of CPU and $m$ units of (RAM) memory are allocated to job $j$. For each machine $i \in [s]$, we denote $C_{g}, M_{g}$ as the \fair allocation of CPU and memory. That is, $C_{g} = C_i/G_i * g_j$ and $M_{g} = M_i/G_i * g_j$.
With a baseline \fair allocation strategy the progress a job makes in each round is equal to $W[C{g}, M{g}]$.


Our objective function is to maximize the throughput. We formulate it as follows using our LP variables.
\begin{equation}
\text{Maximize}  \quad \sum_{j \in J_t}  \sum_{[c,m]} W_j[c, m] \cdot y_{\{c,m,j\}}
\end{equation}
Now, we enforce constraints such that LP solution is feasible in the idealized setting we talked about.
\begincompactitemize
\item Total CPU and memory allocated to jobs is no more than the total capacity available:

\begin{equation}
 \sum_{j \in J_t}  \sum_{[c,m]} c \cdot y_{\{c,m,j\}}  \leq C
\end{equation}

\begin{equation}
 \sum_{j \in J_t}  \sum_{[c,m]} m \cdot y_{\{c,m,j\}}  \leq M
\end{equation}

\item We want the LP to allocate only one configuration of CPU and memory to each job.
\vspace{-0.8em}
\begin{equation}
\forall j \in J_t: \quad  \sum_{[c,m]} y_{\{c,m,j\}}  = 1
\end{equation}
\vspace{-1em}
\item LP solution is atleast as good as the fair allocation.
\vspace{-0.5em}
\begin{equation}
\forall j \in J_t: \quad  \sum_{[c,m]} W_j[c,m] \cdot y_{\{c,m,j\}}  \geq W_j[C_{g}, M_{g}] 
\end{equation}
\vspace{-0.5em}
\end{itemize}

\vspace{-1em}
\begin{theorem}
Throughput achieved by LP(1-5) is at least the throughput achieved by an optimal solution to our problem.\end{theorem}
\vspace*{-1em}
\begin{proof}
Consider an optimal solution $O$ to our problem. Suppose job $j$ receives $c^*$ units of CPU and $m^*$ units of memory in $O$. Then we define the following feasible solution to our LP (1-5): Set $y_{c^*, m^*, j} = 1$.
Clearly, this is a valid solution and satisfies constraints (1-4).
\end{proof}

In our experiments, we solve this as a Integer Linear Program (ILP) where $y_{\{c,m,j\}}$ takes boolean values. For every job, we define the total CPU ($c^*_j$) and memory ($m^*_j$) allocated by the optimal ILP solution as follows.
\vspace{-0.7em}
	\begin{equation}
	\text{For each job j, define}   \quad c^*_j := c  \quad \text{if}  \quad y_{\{c,m,j\} == 1}. 
	\end{equation}
\vspace{-1em}	
		\begin{equation}
	and   \quad m^*_j := m  \quad \text{if}  \quad y_{\{c,m,j\} == 1}. 
	\end{equation}
\vspace{-2em}

\subsubsection{Feasible Allocation on Multiple Machines}
Recall that in the LP(1-5), we assumed that all the resources are present on a single machine. In reality, since these resources are spread across machines, we find a feasible allocation on multiple machines by solving a second LP. The objective here is to minimize the number of jobs that get fragmented to account for the communication overhead when jobs are split across machines. The variables of the second LP are denoted by $x_{i,j}$. Here index $i$ denotes the machine and $j$ denotes the job.
If $x_{i,j} = 1$, it means that resources of job $j$ (that $g_j$ units of GPU, $c^*_j$ units of CPU, and $m^*_j$ units of memory) are allocated on machine $i$. Note that $x_{i,j}$ can be fractional; if so, then job $j$ is split across multiple machines. We can prove that the solution to the second LP ensures that the total number of jobs that get fragmented is at most $3s$. Detailed formulation is in the extended version of the paper~\cite{synergy-arxiv}.  

\subsubsection{Challenges with operationalizing \sysopt}
\label{sec-opt-challenge}
While the allocations identified by \sysopt provides an upper bound on the optimal cluster throughput, it is challenging to materialize these allocations in the real world due to two main reasons;
\begincompactitemize
	\item Solving two LPs per scheduling round is computationally expensive. As cluster size and the number of jobs per round increases, the time to find an optimal allocation increases exponentially (\sref{sec-eval-opt})
	\item The allocation matrix obtained with the second LP can result in fractional GPU allocations when jobs are split across servers; for instance, a valid allocation might assign 3.3 GPUs on server 1 and 2.7 GPUs on server 2 for a 6 GPU job. Realizing such an allocation requires a heuristic rounding off strategy to ensure non-fractional GPU allocations, as GPU time or space sharing, and its impact on job performance is considered beyond the scope of this work.
\end{itemize} 

\subsection{\systune}

We now describe \systune, our heuristic scheduling mechanism. Our goal is to design a scheduling mechanism that performs multi-dimensional resource allocation for DNN jobs, where the GPU demand is fixed, but the auxiliary resource allocations are fungible. In doing so, we want to ensure that (1)  we do not affect the fairness properties of the scheduling policy used,  (2) the expensive GPU resources are not underutilized.

\vheading{Allocation Requirements}. \systune's allocation must satisfy the following requirements.
\begin{itemize}[leftmargin=*,noitemsep,partopsep=0pt,topsep=0.2em,parsep=0pt]
	\item The GPU, CPU, and memory resources requested by a single-GPU job must all be allocated on the same server. 
	\item A multi-GPU distributed-training job can either be consolidated on one machine, or split across multiple machines. In the latter case, \emph{the CPU and memory allocations must be proportional to GPU allocations across servers}. For instance, if a job requires (2GPU, 12 CPU, 300GB DRAM), then while splitting it across two servers, we need to ensure that each server gets (1GPU, 6CPU, 150GB DRAM). This is because, multi-GPU jobs train on a separate process on each GPU, and synchronize at regular intervals, i.e., after one or many iterations. The job performance will vary across processes if each GPU does not get the same ratio of resources, and  will eventually proceed at the speed of the process with the lowest allocation of CPU and memory.
\end{itemize}

 In a multi-tenant cluster, while carving out resources such as CPUs and memory for jobs, it is import to enforce fairness in terms of throughput achieved by individual jobs. We need to ensure that no job runs at a throughput lower than what it would have achieved if we allocated a \fair share of CPU and memory resources. Additionally, we need to respect the priority order of jobs identified by the scheduling policy. For instance, a FIFO scheduling policy can be implemented using a priority queue sorted by job arrival times. 
\systune identifies a set of runnable jobs for a round as the top \textit{n} jobs from the scheduling queue, whose GPU demands can be exactly satisfied by the available servers in the cluster. 
 \systune picks this runnable job set irrespective of the job's other resource demands - which are fungible. Note that, unlike \sysgreedy, we do not skip over any jobs unless it cannot be scheduled (GPU demand cannot be met). Therefore, we never underutilize the GPUs when the cluster is at full load.

Next, \systune greedily packs each of these runnable jobs along multiple resource dimensions on one of the available servers, with the objective of minimizing fragmentation. To achieve this, \systune sorts the runnable jobs by their GPU demands, followed by CPU, and memory demand.  For each job $j$ in order, \systune then picks the server with the least amount of free resources just enough to fit the demand vector of $j$. If it is a multi-GPU job, then we find a minimum set of servers with sufficient GPU availability that can fit the job's demands in entirety. However, it is possible that the job cannot fit in the cluster along all dimensions. In such a case, 
	\begin{enumerate}[wide, labelwidth=!, labelindent=1pt, itemsep=0.5pt, topsep=0pt]
		\item We check if the job's demand vector is greater than proportional share of resources, In this case, we switch the job's demand to \fair share and retry.
		\item If the job still does not fit the cluster, or if the job's demand vector was less than or equal to GPU proportional allocation in step (1), then, we do the following.
		\begin{enumerate}[labelwidth=!, labelindent=1pt, itemsep=1pt, topsep=1pt]
			\item We repeat step (1) ignoring the job's CPU and memory requirements. We find a server that can just satisfy the job's GPU requirements. We know by construction that there is atleast one job on this server, which is allocated more than \fair share of resources. We identify the job or a set of jobs ($J_s$) on this server by switching whom to \fair share, we can release just as much resources required by the current job $j$. We switch the jobs in $J_s$ to fair-share and by design, job $j$ will fit this server.
			\item We continue this recursively for all runnable jobs.
		\end{enumerate}
	\end{enumerate}

In the worst case, all the running jobs in a round could be allocated \fair share of resources.  Therefore, \sysname ensures that its allocations results in job throughputs that are never worse than \fair allocation. In \sref{sec-eval-opt}, we empirically compare \systune to \sysopt showing that it is practical and near-optimal.

\subsection{Implementation}
\label{sec-impl}

We implement \sysname and an associated simulator in Python. 
Our scheduler is event-driven. There is a global event queue where job arrivals, schedule events, and deploy events are queued. These events are handled in the order of their arrival time. There is a priority job queue, where all the jobs arriving into the cluster are added, post profiling. This queue is sorted by the priorty metric decided by the scheduling policy; for instance, SRTF sorts the jobs in the order of job remaining time.

 When a schedule event occurs, the scheduler collects a list of runnable jobs from the job queue and identifies the appropriate placement for these jobs for the following round, either using \sysgreedy, \systune or \sysopt. Then when a deploy event is triggered, these allocations are deployed on to the cluster. By default, every job requests for a lease update to continue running on the same server~\cite{gavel}. The scheduler then either grants a lease update or terminates the lease for the job, adding it back to the job queue.
 
 The scheduler and the DNN jobs interact via a thin API provided by the \sysname data iterator. DNN job scripts must be updated to call the \sysname iterator which is a wrapper around the default PyTorch~\cite{pytorch} and DALI~\cite{dali} iterators. The iterator handles registering the job with the scheduler, and appropriately sending lease updates. It also checkpoints the job to a shared storage if its lease is terminated. The iterator also synchronizes across GPU processes for a multi-GPU job to ensure that each process makes identical progress. We use \texttt{gRPC}~\cite{grpc} to communicate between the scheduler and the jobs. 
 
 We implement \sysopt in \texttt{cvxpy}~\cite{cvxpy} for use in our simulator. The optimistic profiling module is also implemented in Python, and it profiles the incoming jobs hooked to the \sysname iterator, prior to the job's initial addition to the scheduling queue (a one time overhead for each job).
	
\section{Evaluation}
\label{sec-eval}
In this section, we use trace-driven simulations from production cluster traces, and physical cluster deployment to evaluate the efficacy  of \sysname. Our evaluation seeks to answer the following questions.

\begin{table}[!t]
	\small
	\centering
	\ra{0.9}
	
	\begin{tabular}{l@{\hskip10pt}r@{\hskip15pt}l@{\hskip5pt}r@{\hskip25pt}r@{\hskip10pt}}
		\toprule[1.2pt]
		Task & Model & Dataset\\ 
		\midrule
		\multirow{5}{*}{\shortstack[l]{Image}}  & Shufflenetv2~\cite{zhang2018shufflenet} & 	\multirow{5}{*}{\shortstack[l]{ImageNet~\cite{russakovsky2015imagenet}}} \\
		& AlexNet~\cite{krizhevsky2012imagenet} &  \\
		& Resnet18~\cite{he2016deep} &  \\
		& MobileNetv2~\cite{sandler2018mobilenetv2} & \\ 
		& ResNet50~\cite{he2016deep} & \\
		\midrule
		\multirow{3}{*}{\shortstack[l]{Language}} &  GNMT~\cite{wu2016google} & WMT16~\cite{wmt16} \\
				 &  LSTM~\cite{lstm} & Wikitext-2~\cite{wikitext_2} \\
				 &  Transformer-XL~\cite{transformer-xl} & Wikitext-103~\cite{wikitext_2} \\
		\midrule
		\multirow{2}{*}{\shortstack[l]{Speech}} & M5~\cite{dai2017very} & Free Music~\cite{defferrard2016fma} \\
				& DeepSpeech~\cite{deepspeech} & LibriSpeech~\cite{librispeech} \\
		\bottomrule[1.2pt]
	\end{tabular}
    \vspace{-1em}
	\mycaption{Models used in this work}{}
	\label{tbl-dataset-models}
	\vspace{-2em}
\end{table}

\begincompactitemize
	\item Does \sysname's resource-sensitive scheduling improve cluster objectives such as makespan and average JCT in a physical cluster (\sref{sec-eval-deploy}) and in trace-driven simulations of large-scale clusters (\sref{sec-eval-sim-all}) ?
	\item How does \systune and \sysgreedy perform with different workload splits and how well do they utilize available resources  (\sref{sec-eval-split})?
	\item How does \sysname perform on different CPU:GPU ratios (\sref{sec-eval-cpu-cores})?
	\item Compare \systune to \sysopt (\sref{sec-eval-opt})?
	\item Compare \sysname to big data schedulers (\sref{sec-eval-big-data})?
	
\end{itemize}

\subsection{Experimental setup}
\label{sec-eval-setup}
\vheading{Clusters}. Our experiments run on both a physical and a large simulated, homogeneous cluster.  Our experiments are performed on state-of-the-art internal servers at Microsoft - these servers are part of a larger multi-tenant cluster. We run physical cluster experiments on a cluster with 32 V100 GPUs across 4 servers. Each server has 500GB DRAM, 24 CPU cores, and 8 GPUs.
Unless otherwise specified, our experiments assume a CPU:GPU ratio of 3  and fair-share memory allocation of 62.5GB per GPU, matching the server configurations above.  For simulations, we assume two cluster sizes; a 128 GPU cluster across 16 servers and a 512 GPU cluster across 64 machines, where each machine resembles the physical server configuration mentioned above.

\vheading{Models}. Our experiments consider 10 different DNNs (CNNs, RNNs, and LSTMs) as shown in Table~\ref{tbl-dataset-models}. We categorize these models by task (image, language, and speech) and assign a certain weight to these tasks in our traces. We call this a workload \textit{split}. For instance, if the split for a given trace is (30,40,30), then the percentage of image, language, and speech models in the job trace is 30\%, 40\% and 30\% respectively. All experiments are performed on PyTorch 1.1.0. 

\setlength{\tabcolsep}{4pt}
\begin{table}[!t]
  \small
  \centering
  \ra{0.9}
 \begin{tabular}{ccccccc}
	\toprule[1pt]
	      \textbf{Policy}  &   \textbf{Workload} & \multirow{2}*{\textbf{Mechanism}}  & \multicolumn{2}{c}{\textbf{Time (hrs)}}\\ 
	      \textbf{(Metric)} & \textbf{Split} & & \textbf{Deploy} & \textbf{Simulate}  \\
	      \midrule
         \multirow{2}*{\shortstack{FIFO \\(Makespan)}} & \multirow{2}*{60-30-10}    & Proportional   & 16 & 15.67 \\
         &	&                         Tune & 11.6 & 11.33 \\
         &	&                         Opt & - & 11.01 \\
         \midrule
         \multirow{2}*{\shortstack{SRTF \\(Avg JCT)}} & \multirow{2}*{30-60-10} &    Proportional & 4.81  & 4.52\\
         &	& Tune & 3.21 &  3.19\\
          &	& Opt & - & 3.06 \\
          \midrule
          \multirow{2}*{\shortstack{SRTF \\(99 Percentile\\ JCT)}} & \multirow{2}*{30-60-10} &    Proportional & 17.32  & 16.85\\
          &	& Tune & 8.59 &  8.54\\
          &	& Opt & - & 8.21 \\
	\bottomrule[1pt]
   \end{tabular}
   \vspace{-1em}
       \mycaption{Physical cluster experiments}{This table compares the makespan, average JCT, and 99th percentile JCT for two different traces; (1) a static trace using FIFO  (2) a dynamic trace using SRTF. \systune improves makespan by 1.4\myx, average JCT by 1.5 \myx and 99th percentile JCT by 2\myx.}
   \vspace{-1em}
\label{tbl-eval-dep}
\end{table}

\vheading{Traces}. We run our physical and simulated experiments using publicly available production traces from Microsoft Philly cluster~\cite{philly}.
We show evaluation with the actual Philly trace preserving the job GPU demand, arrival time, and duration, on a cluster of 512 GPUs in ~\sref{sec-eval-philly}. We use a subrange of the trace containing 8000 jobs. 

However, to comprehensively evaluate how \sysname reacts to varying cluster load, workload composition, and job duration, for all other experiments, we construct a production-derived trace as follows: we extract job GPU demand from the Philly trace and assign a model based on the chosen \jsplit. We then appropriately scale the job runtime and arrival time for the chosen cluster size, while keeping the job duration distribution similar to the one in Philly trace as follows:
\begincompactitemize
    \item \textbf{Duration}. The duration of each job for the baseline \fair allocation is sampled from an exponential distribution: the job duration is set to $10^x$ minutes, where x is drawn uniformly from [1.5,3] with 80\% probability, and from [3,4] with 20\% probability similar to the trace duration used in prior work~\cite{gavel}.  
    \item \textbf{Arrival}. We classify derived traces into two kinds based on the job arrival time : (1)  a \textit{static} trace where all the jobs arrive at the start of the workload, and (2) a \textit{dynamic} trace, where the job arrival time is determined by load, a Poisson distribution at a rate $\lambda$. 
\end{itemize}
 


The derived traces with varying job arrival rates uses a 128 GPU cluster. In both cases, we report the average metrics such as JCT across a set of 1000 jobs in steady state.

\begin{figure*}[!t]
\centering
\subfloat[Average JCT with \sysname \label{tbl-philly-eval}]{
  \ra{1}
 \begin{tabular}[b]{cccc}
	\toprule
	& \multicolumn{3}{c}{\textbf{Avg JCT(hrs)}}\\
		\textbf{Policy} & \textbf{SRTF} & \textbf{LAS} & \textbf{FIFO} \\
	\midrule
	\textbf{GPU-prop.} & 30 & 32 & 71 \\
	\textbf{\sysname} & 26 & 28 & 62 \\
	
	\bottomrule
   \end{tabular}
}
\qquad
\subfloat[Cluster metrics (SRTF)\label{fig-sim-srtf-cdf-tbl}]{
\ra{0.7}
    \begin{tabular}[b]{cccc}
    \toprule
      \multicolumn{2}{c}{\textbf{JCT (hrs)}} & \textbf{Short} & \textbf{Long} \\
      \midrule
      \multirow{2}{*}{\textbf{Avg}} & Prop. & 2 & 80 \\
      & Synergy & 1.7 & 68\\
    \midrule
      \multirow{2}{*}{\textbf{99p}} & Prop. & 9 & 660\\
      & Synergy & 4 & 641\\
      \bottomrule
    \end{tabular}
}
\qquad
\subfloat[JCT speedup across jobs \label{fig-sim-srtf-impr}]{{\includegraphics[width=.25\textwidth]{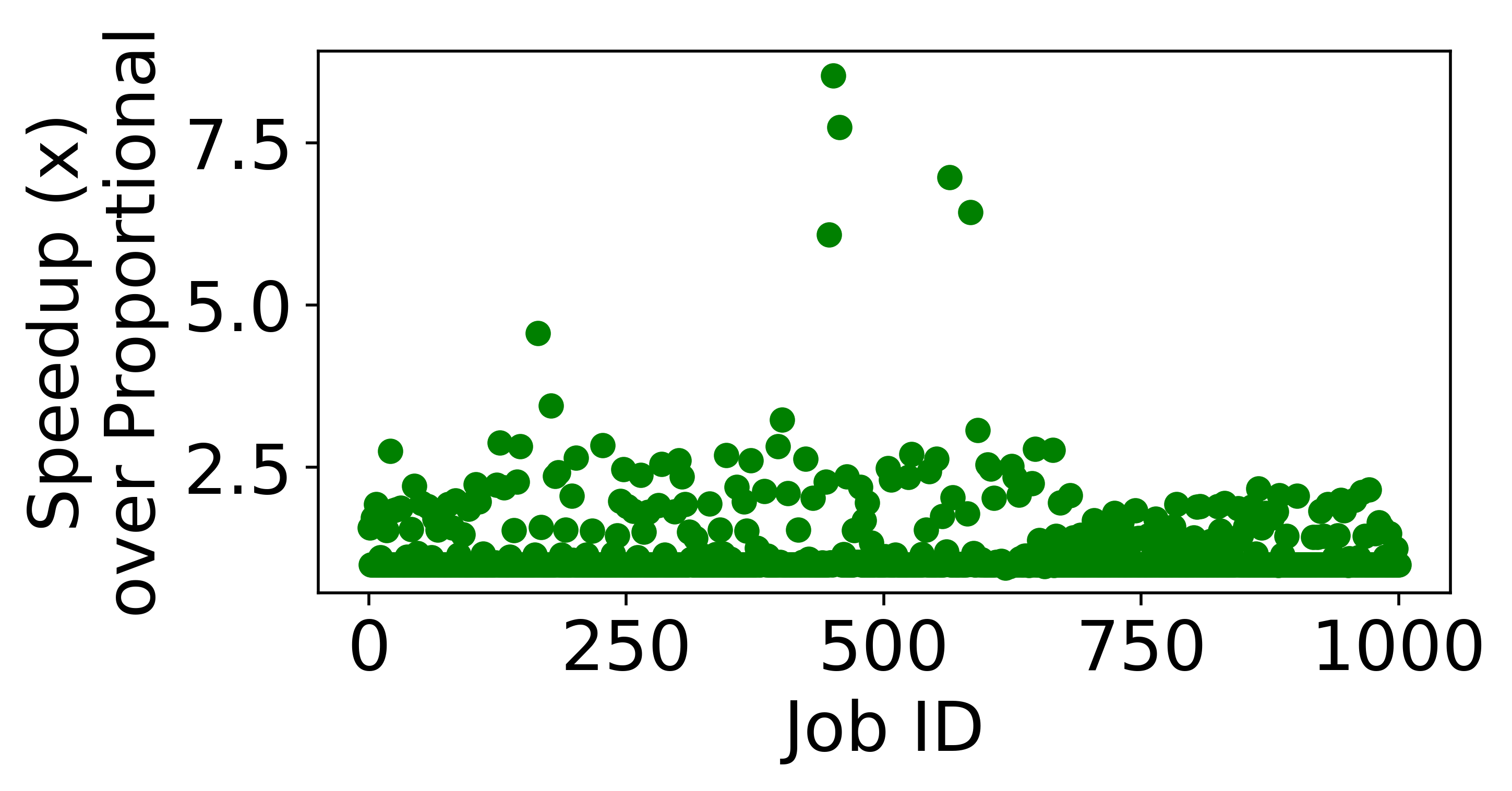} }}%
\vspace*{-2em}
\mycaption{Evaluation on Philly Trace}{On a real production trace, \sysname improves avg JCT across a range of scheduling policies over \fair scheduling. The JCT of individual jobs improves by upto 9\myx with \sysname.}
\vspace*{-1em}
\label{fig-eval-philly}
\end{figure*}
\begin{figure*}[!htb]
		\centering
	\subfloat[LAS (multi) \label{fig-sim-las-jct}]{{\includegraphics[width=.3\textwidth]{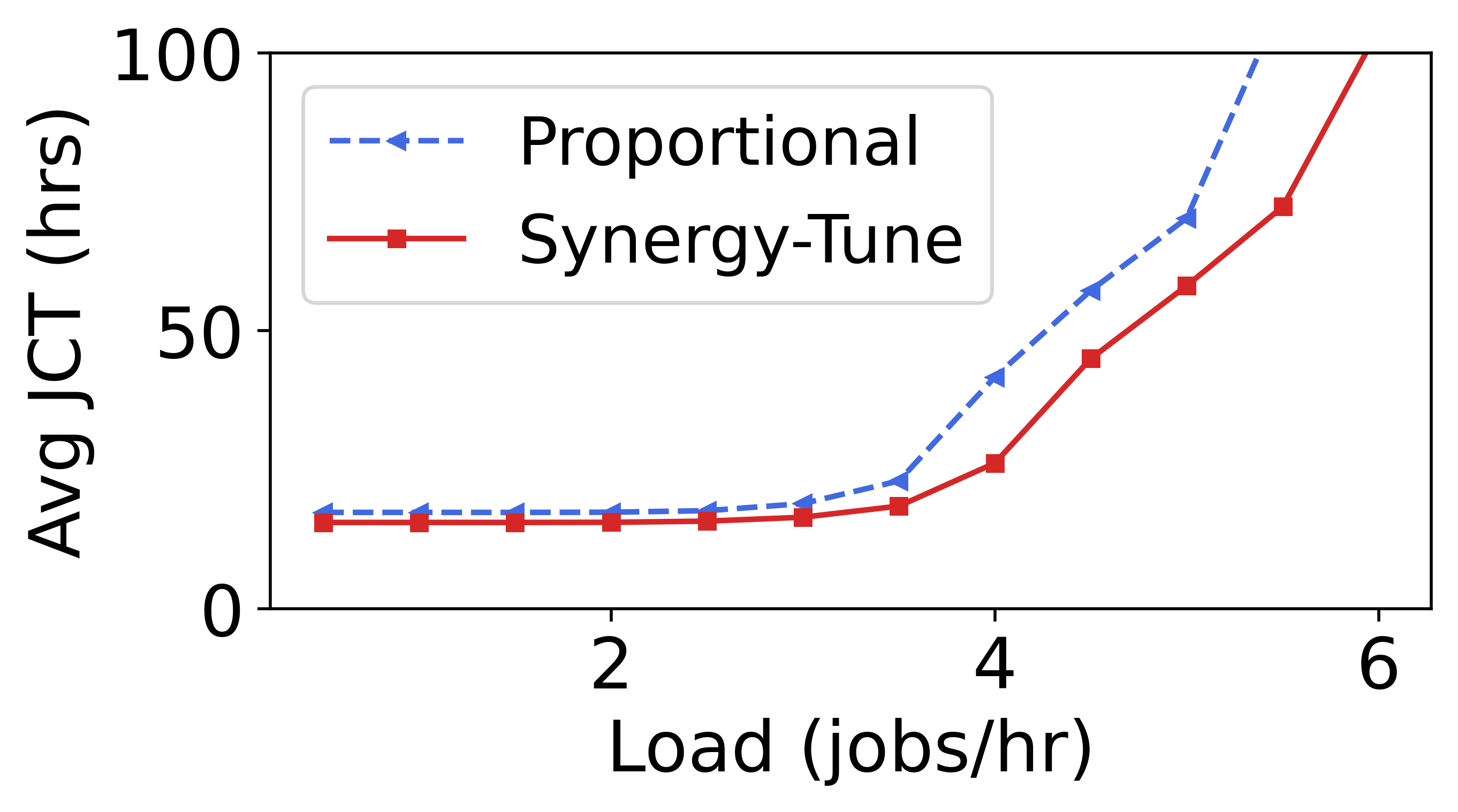} }}%
	\qquad
	\subfloat[CDF of JCT at load 4 (short) \label{fig-sim-las-cdf-short}]{{\includegraphics[width=.3\textwidth]{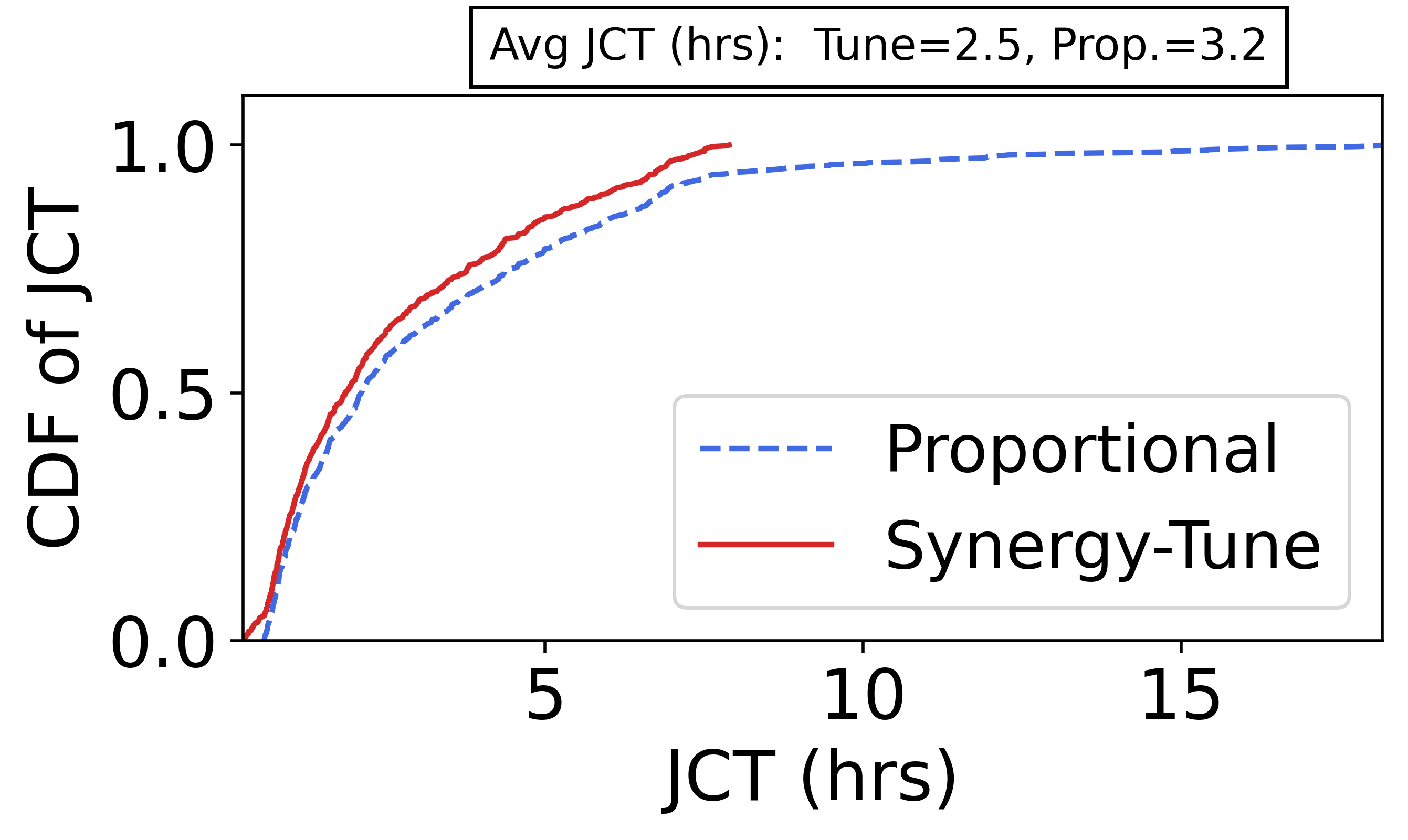} }}%
	\qquad
	\subfloat[CDF of JCT at load 4 (long) \label{fig-sim-las-cdf-long}]{{\includegraphics[width=.3\textwidth]{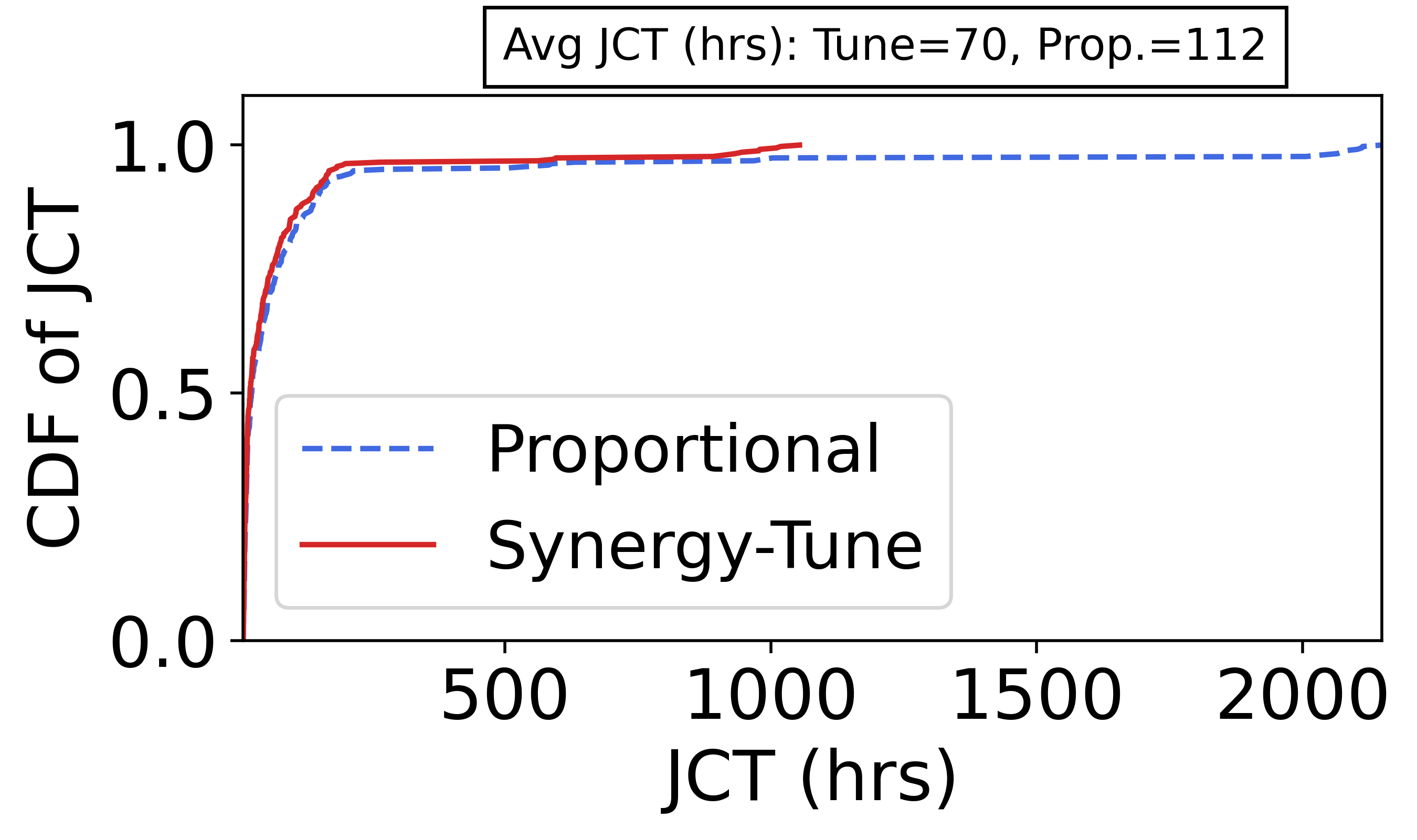} }}%

\vspace{-1em}
\mycaption{Average JCT and CDF of long and short jobs for LAS policy}{}
\vspace*{-1em}
\label{fig-eval-sim-las}
\end{figure*}
\begin{figure*}[!htb]

\subfloat[\revised{SRTF (multi)} \label{fig-sim-srtf-jct}]{{\includegraphics[width=.3\textwidth]{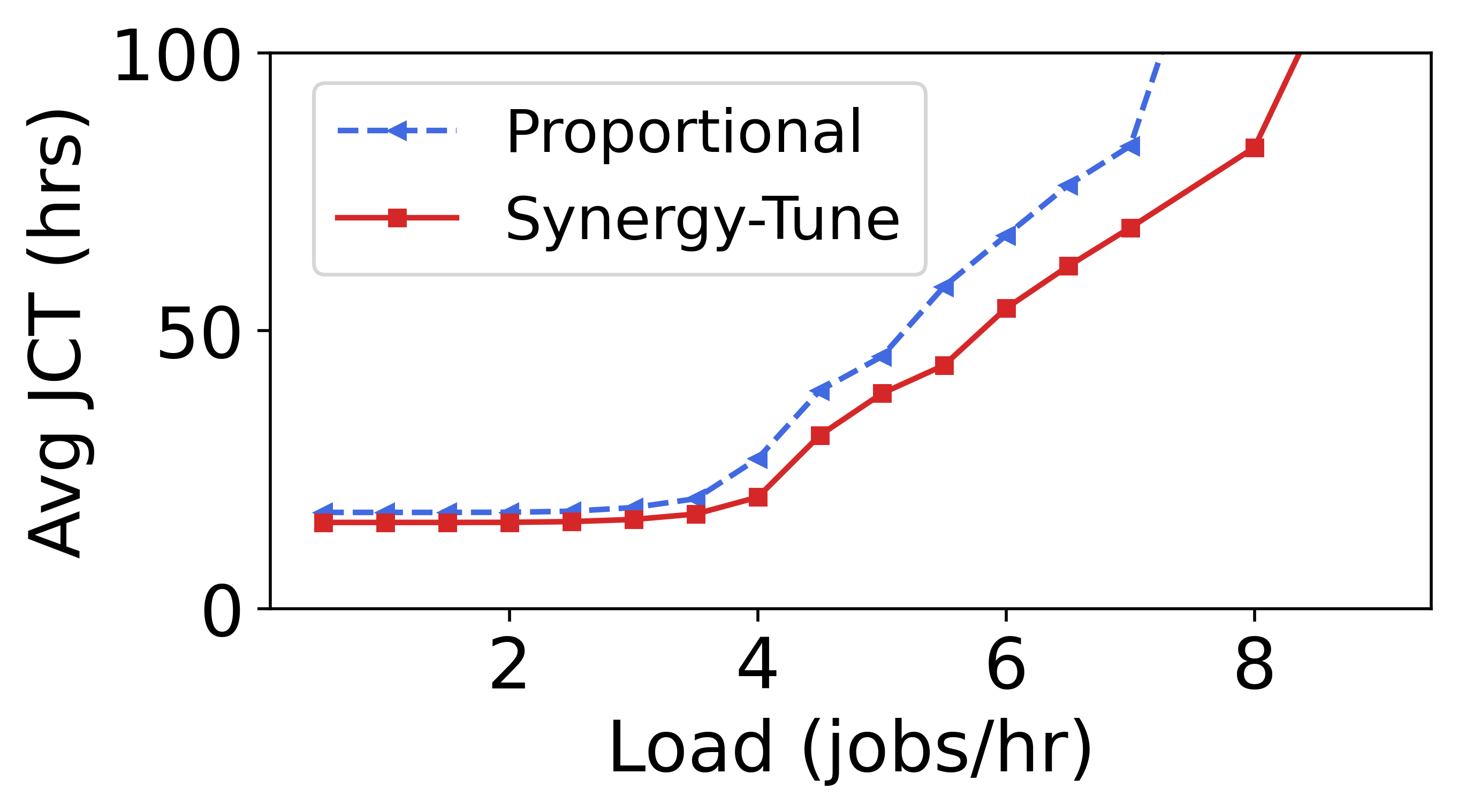} }}%
\qquad
\subfloat[\revised{CDF of JCT at load 5.5 (short)} \label{fig-sim-srtf-cdf-short}]{{\includegraphics[width=.31\textwidth]{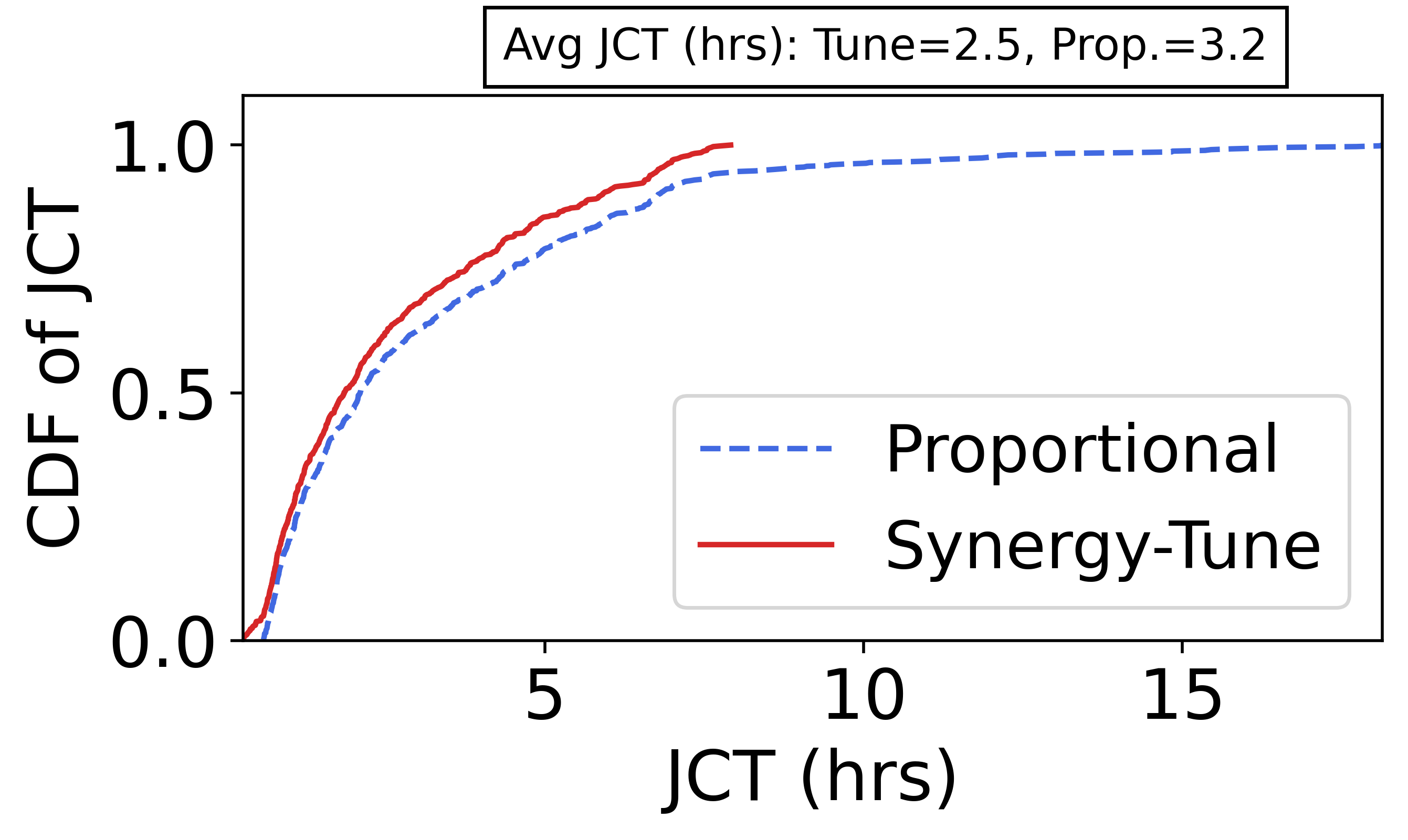} }}%
\qquad
\subfloat[\revised{CDF of JCT at load 5.5 (long)} \label{fig-sim-srtf-cdf-long}]{{\includegraphics[width=.3\textwidth]{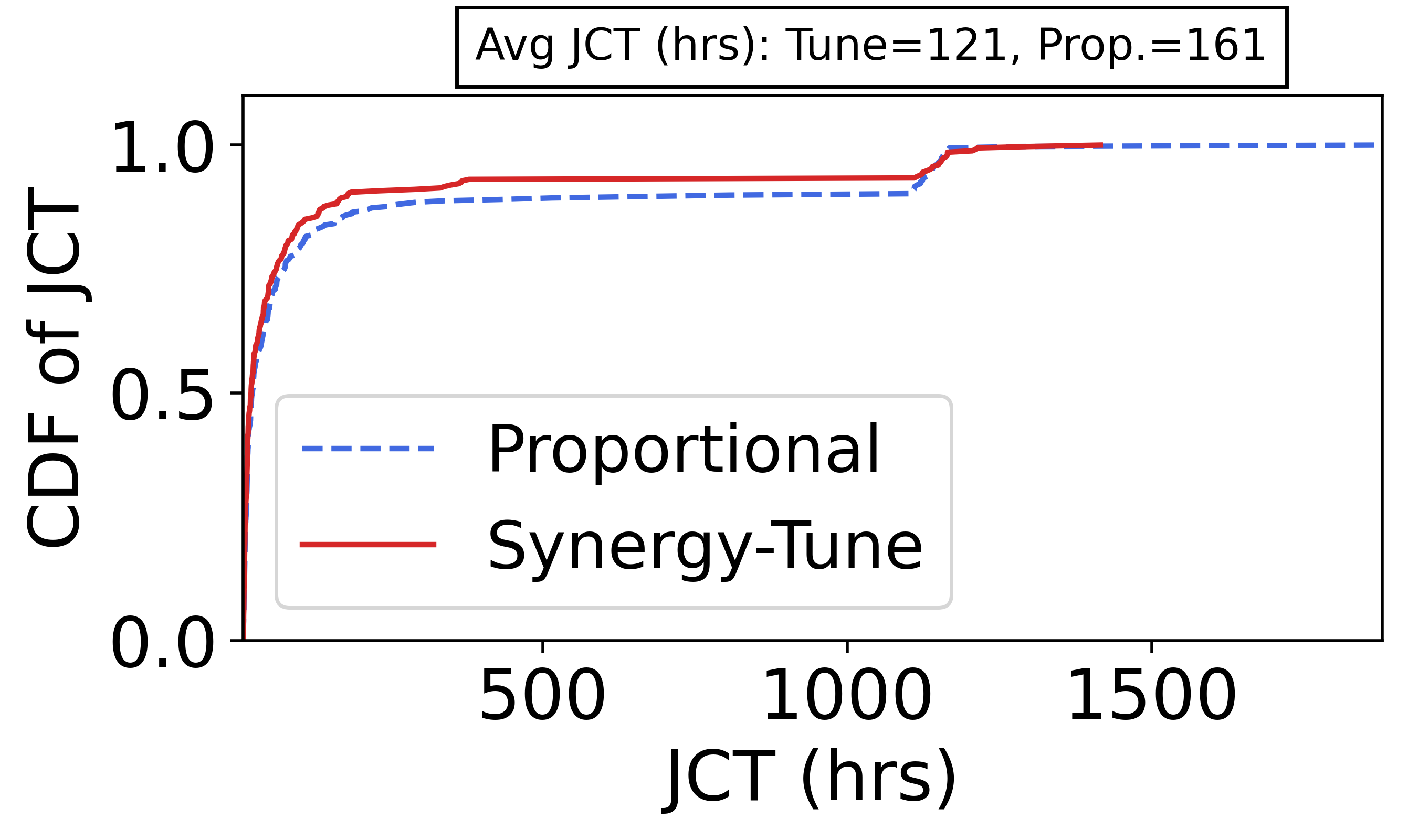} }}%
\vspace*{-2em}
\mycaption{\revised{Average JCT and CDF of long and short jobs for SRTF policy}}{}
\vspace*{-1em}
\label{fig-eval-sim-srtf}
\end{figure*}

For the physical cluster experiment, we choose a fixed arrival rate for the derived trace that keeps our cluster at \textit{full load} (GPU demand of all runnable jobs > available GPUs in the cluster). For the simulated experiments, we vary the load $\lambda$ on the cluster to evaluate its impact on cluster metrics. For the simulated experiments, we show results for two trace categories - (1) all jobs request single-GPU (2) multi-GPU distributed training jobs that request upto 16 GPUs.

\vheading{Policies and metrics}. We evaluate \sysname against \fair scheduling for 4 different scheduling polices; FIFO, SRTF, LAS, and FTF. For a static trace, we measure makespan (time to complete all jobs submitted at the beginning of the trace) and for the dynamic job traces, we measure the average job completion time (JCT) of a subset of jobs in steady state (cluster at full load), and their CDF. 

\subsection{End-to-End Physical Cluster Experiments}
\label{sec-eval-deploy}
For the physical cluster experiments, we run a \systune (\textit{tune}) and \fair allocation (\textit{proportional}) for two different workload traces. 
(1) A static production-derived trace of 100 jobs with a \jsplit (60,30,10), scheduled using FIFO and evaluated for makespan. (2) A dynamic production-derived trace with continuous job arrivals and a split of (30,60,10), scheduled using SRTF and evaluated for average and 99th percentile JCT. Both scenarios use an appropriately sized trace that keeps the cluster fully loaded. We compare the obtained results to that of the simulator by replaying the same trace. Additionally, we compare our metrics to the upper bound generated by the optimal solution, \sysopt (\textit{opt}). The results are shown in Table~\ref{tbl-eval-dep}.

\systune reduces the makespan of static trace by 1.4\myx when compared to \fair allocation. For the dynamic trace, \systune reduces average JCT of steady-state jobs by 1.5\myx while reducing the 99th percentile JCT of these jobs by 2\myx as shown in Table~\ref{tbl-eval-dep}.

We compare the observed results from physical experiments to the same trace replayed on our simulator. As shown in Table~\ref{tbl-eval-dep}, 
the difference between metrics in real and simulated clusters are less than 5\%, demonstrating the fidelity of the simulator.
We also see from Table~\ref{tbl-eval-dep} that the cluster objectives achieved by \systune are within 4\% of the optimal solution in this case. We do not deploy the optimal allocations due to the challenges enumerated in \sref{sec-opt-challenge}

\subsection{End-to-end results in simulation}
\label{sec-eval-sim-all}

\begin{figure}[!t]
	\centering
	\vspace{-1em}
	\subfloat[FIFO (single) \label{fig-sim-fifo-jct}]{{\includegraphics[width=.22\textwidth]{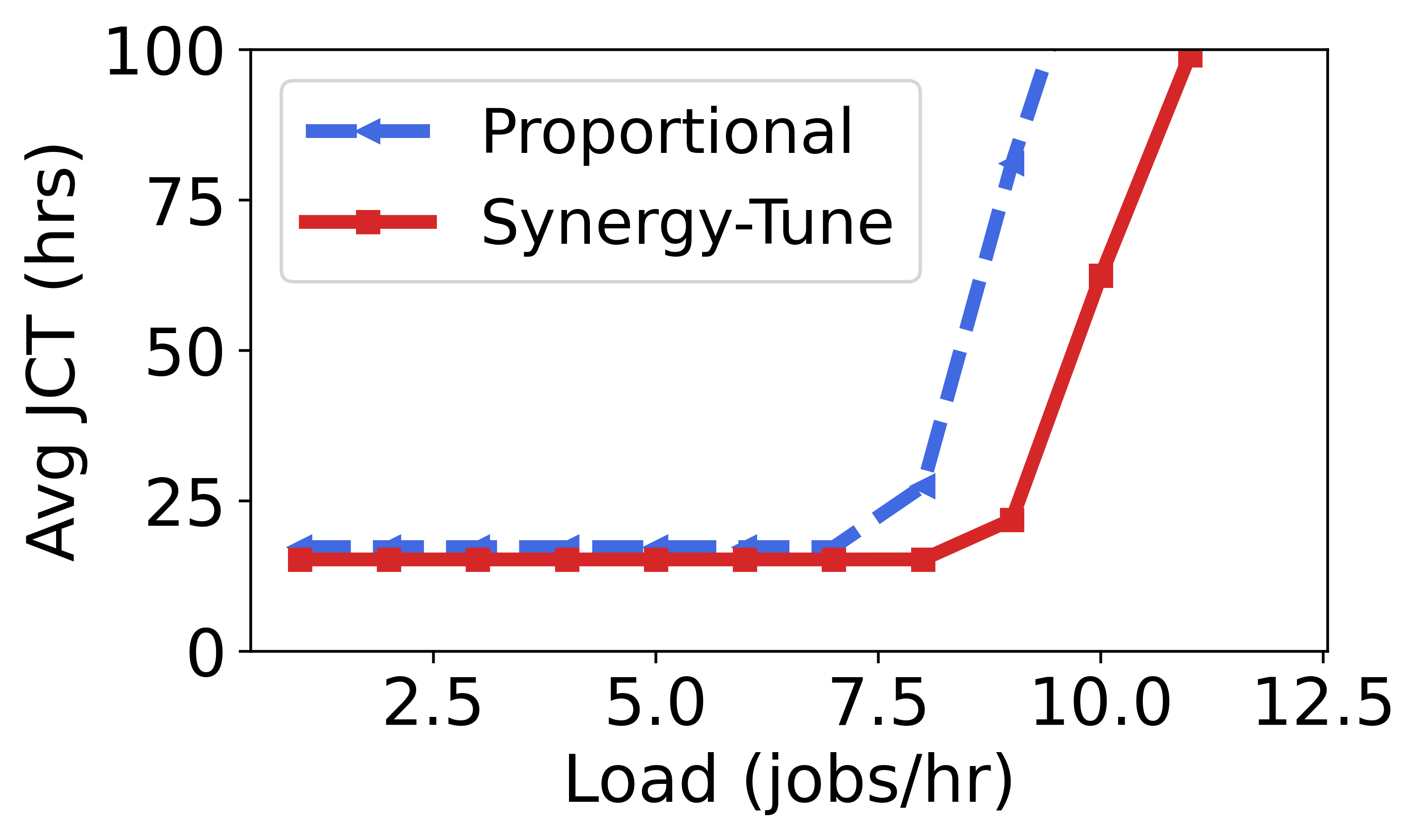} }}%
	\quad
	\subfloat[CDF (9 jobs/hr) \label{fig-sim-fifo-cdf}]{{\includegraphics[width=.22\textwidth]{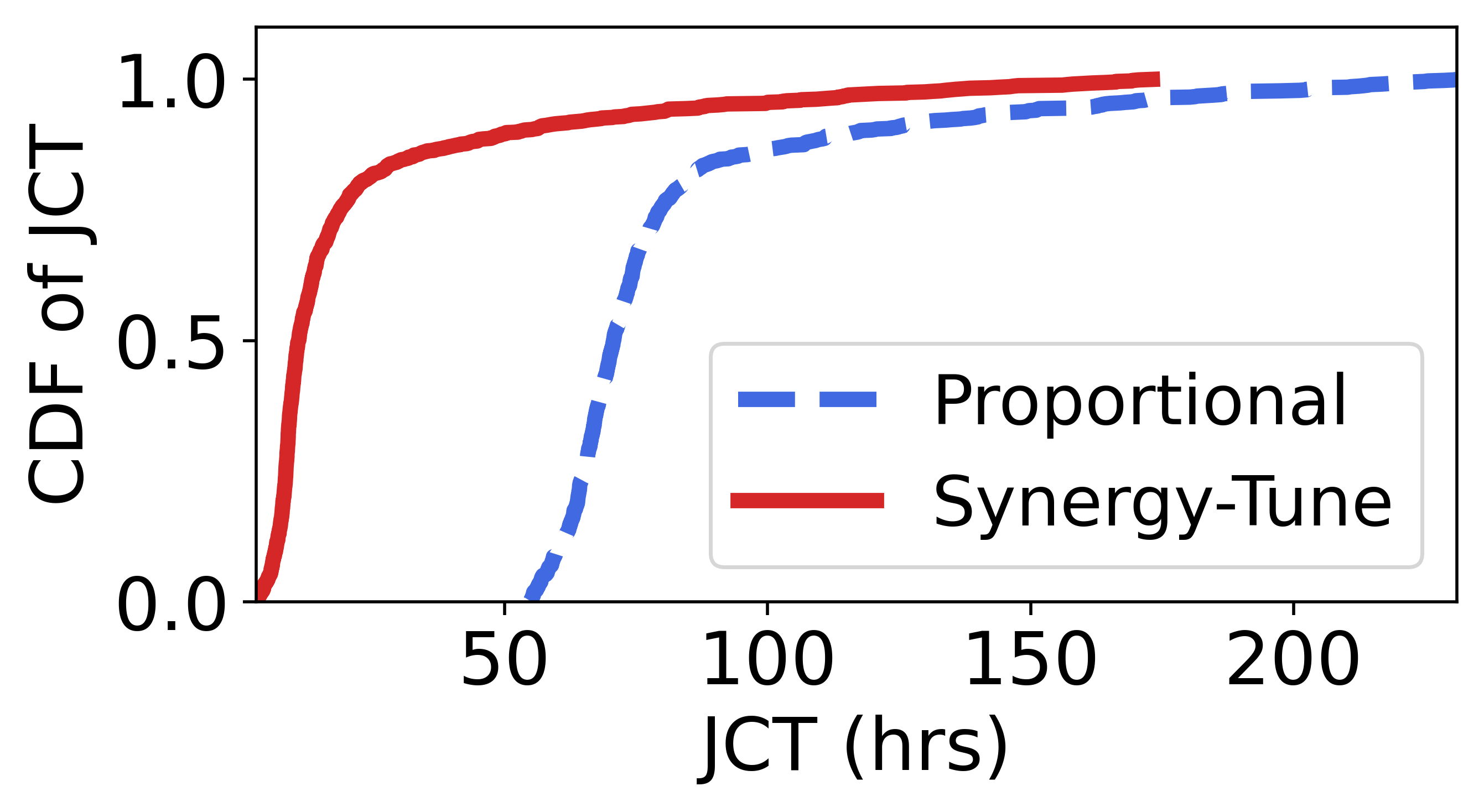} }}%

	\mycaption{Average JCT and CDF for FIFO}{\sysname improves the average JCT significantly compared to \fair allocation for varying cluster load. At a load of 9 jobs/hr, \sysname reduces average JCT from 81hrs to 22hrs, which is close to the upperbound of 20hrs predicted by \sysopt. }

	\label{fig-eval-sim-fifo}
\end{figure}

\subsubsection{Simulation with production traces}
\label{sec-eval-philly}
We run simulated experiments on a cluster of 512 GPUs across 64 servers using a subrange of the publicly available Philly trace published by Microsoft~\cite{philly}.  We assume a workload split of (20,70,10) for this trial.  Table~\ref{tbl-philly-eval} lists the average JCT with \sysname and GPU-proportional scheduling for three different scheduling policies. Across all policies, \sysname is able to reduce the average JCT compared to GPU-proportional scheduling due to better split of resources between jobs. \revised{The gains in \sysname can be attributed to reallocating the underutilized resources from a job to a different, resource-sensitive job whose throughput can improve with the increased allocation.}

We show a detailed overview of the average and 99th percentile JCT for SRTF policy in Table~\ref{fig-sim-srtf-cdf-tbl}.
We split the set of 1000 monitored jobs into short (JCT $\textless$ 4 hrs) and long jobs. \sysname reduces the tail of the distribution by 2.2\myx for short jobs and the average JCT of both long and short jobs by 15\%. For each of the 1000 monitored jobs, we plot the individual job speedup with respect to \fair scheduling in Figure~\ref{fig-sim-srtf-impr}. We see that \sysname speeds up jobs by upto 9\myx using better resource allocations.

\subsubsection{Simulation with varying load}
\label{sec-eval-sim}
We run simulated experiments on a cluster of 128 GPUs across 16 servers using production-derived traces. 
We evaluate \sysname against \fair allocation mechanism for 4 different scheduling policies - FIFO, SRTF, LAS and FTF. We run dynamic workload traces, where jobs arrive continuously at a rate governed by a Poisson distribution. 
We show results for both single-GPU traces (where all jobs request 1 GPU) and multi-GPU traces (where jobs request upto 16 GPUs). \revised{Our metric of evaluation is the average JCT of a set of 1000 jobs in cluster steady state.}

We show the results for three scenarios : LAS (multi-GPU trace) in Figure~\ref{fig-eval-sim-las}, \revised{SRTF (multi-GPU trace) in Figure~\ref{fig-eval-sim-srtf}}, and  FIFO (single GPU trace) in Figure~\ref{fig-eval-sim-fifo}. 
 In all cases, we assume a workload split of (20,70,10). We plot both average JCT and the CDF of job completion times for a specific cluster load in both scenarios. For the multi-GPU trace, we split the CDF into those for short and long jobs to distinctly differentiate the tail of the distribution. We make three key  observations.


\begin{figure}[!t]
	\centering
	\vspace*{-1em}
 \subfloat[GPU utilization\label{fig-eval-sim-gpufrag}]{{\includegraphics[width=.21\textwidth]{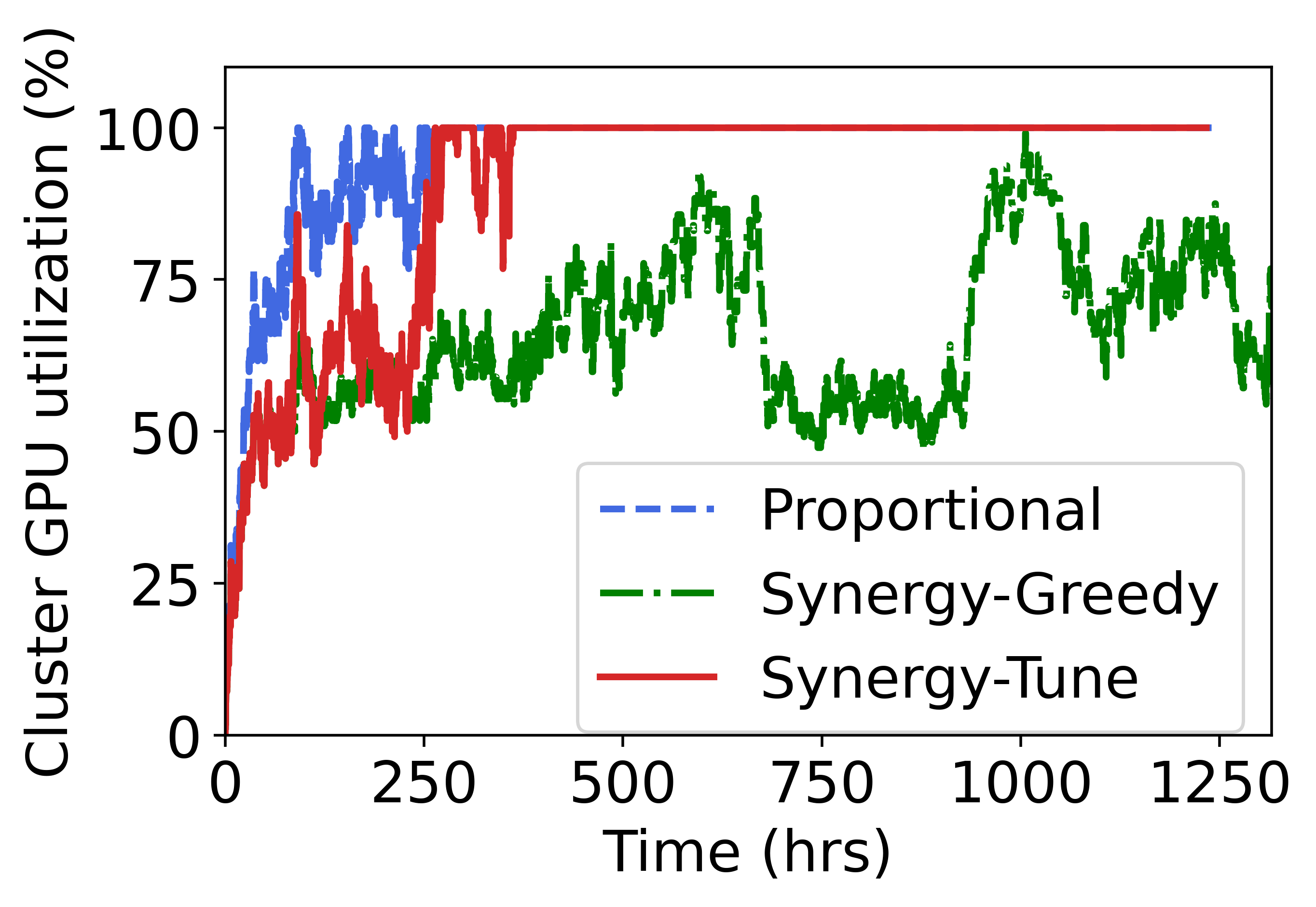} }}%
\quad
\subfloat[CPU utilization\label{fig-eval-cpu-util}]{{\includegraphics[width=.21\textwidth]{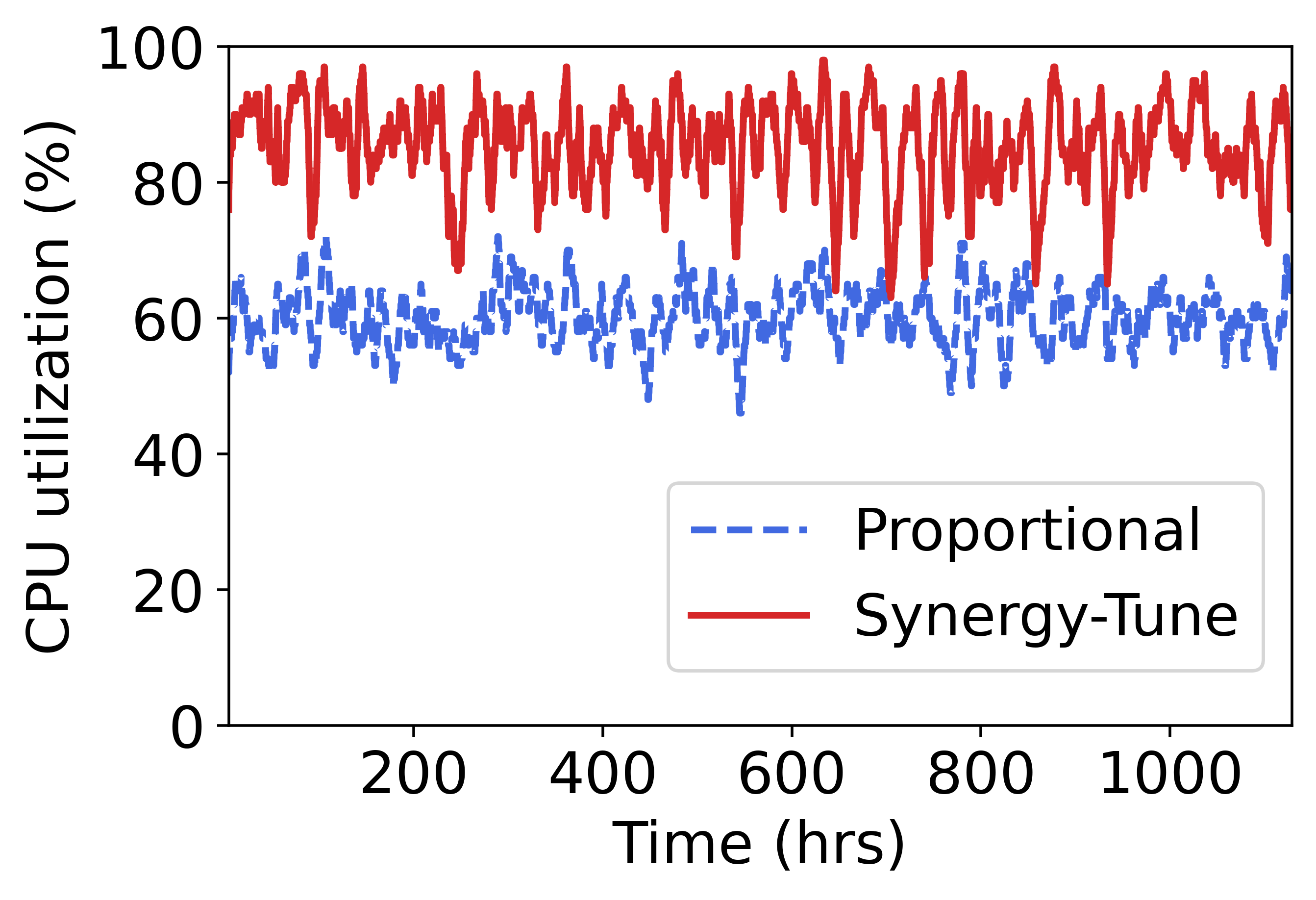} }}
	
 \vspace*{-1em}
 \caption{Cluster resource utilization}
 \vspace*{-1em}
\label{fig-util}
\end{figure}

First, \systune improves average JCT by \upto 3.4\myx in the single-GPU trace, and \upto 1.6\myx in the multi-GPU trace by speeding up resource sensitive jobs with disproportionate allocation.  The improvement in average JCT is higher as the load increases, because at low load the cluster is not at full capacity. As load increases, jobs start to get queued and incur queuing delay before being scheduled on the cluster. Since \sysname significantly speeds up individual jobs using disproportionate resource allocation, pending jobs can get scheduled faster, thereby reducing their queuing delays. Therefore \sysname improves cluster metrics by both reducing qeuing delays and speeding up individual jobs. \revised{Note that, in \fair allocation, at higher loads, all CPUs and memory in the system are allocated to the running jobs but they can still be underutilized by individual jobs.  We show later in Figure~\ref{fig-eval-cpu-util}, how \sysname's resource-sensitivity aware allocation improves CPU utilization in the system compared to \fair allocation. At low load, jobs are spread across the cluster and the unallocated CPU and memory is assigned to the jobs that benefit from additional auxiliary resources.}
Second, \systune is able to sustain a larger cluster load than \fair allocation. For multi-GPU scheduling with LAS, \systune reduced the 95th percentile JCT of long jobs by 2\myx. Third, the average JCT achieved with \systune is within 10\% of the optimal solution in all cases.

\begin{figure*}[!t]
	\centering
 \subfloat[Split=(20,70,10)\label{fig-sim-fifo-s1}]{{\includegraphics[width=.3\textwidth]{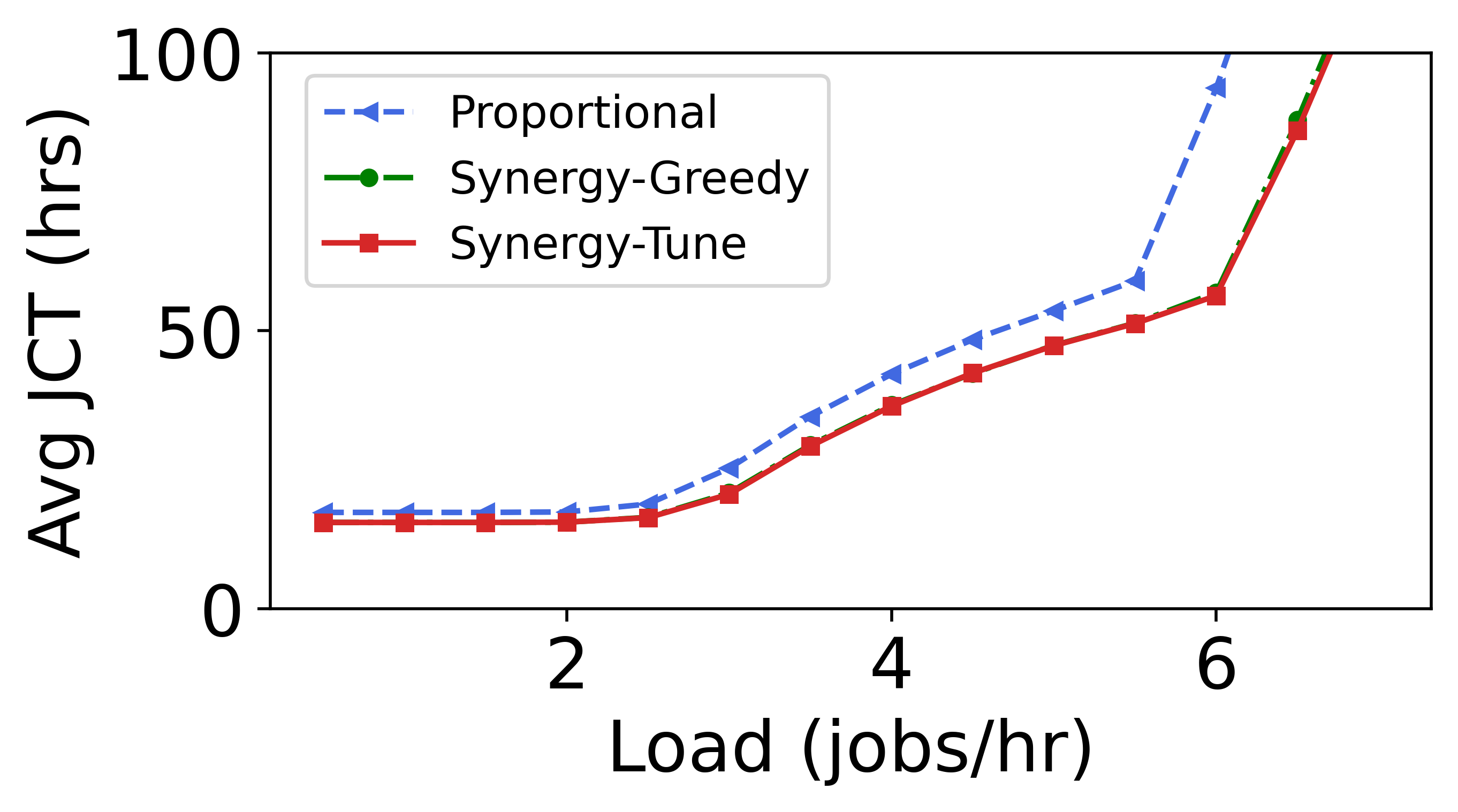} }}%
\qquad
\subfloat[Split=(33,33,33)\label{fig-sim-fifo-s2}]{{\includegraphics[width=.3\textwidth]{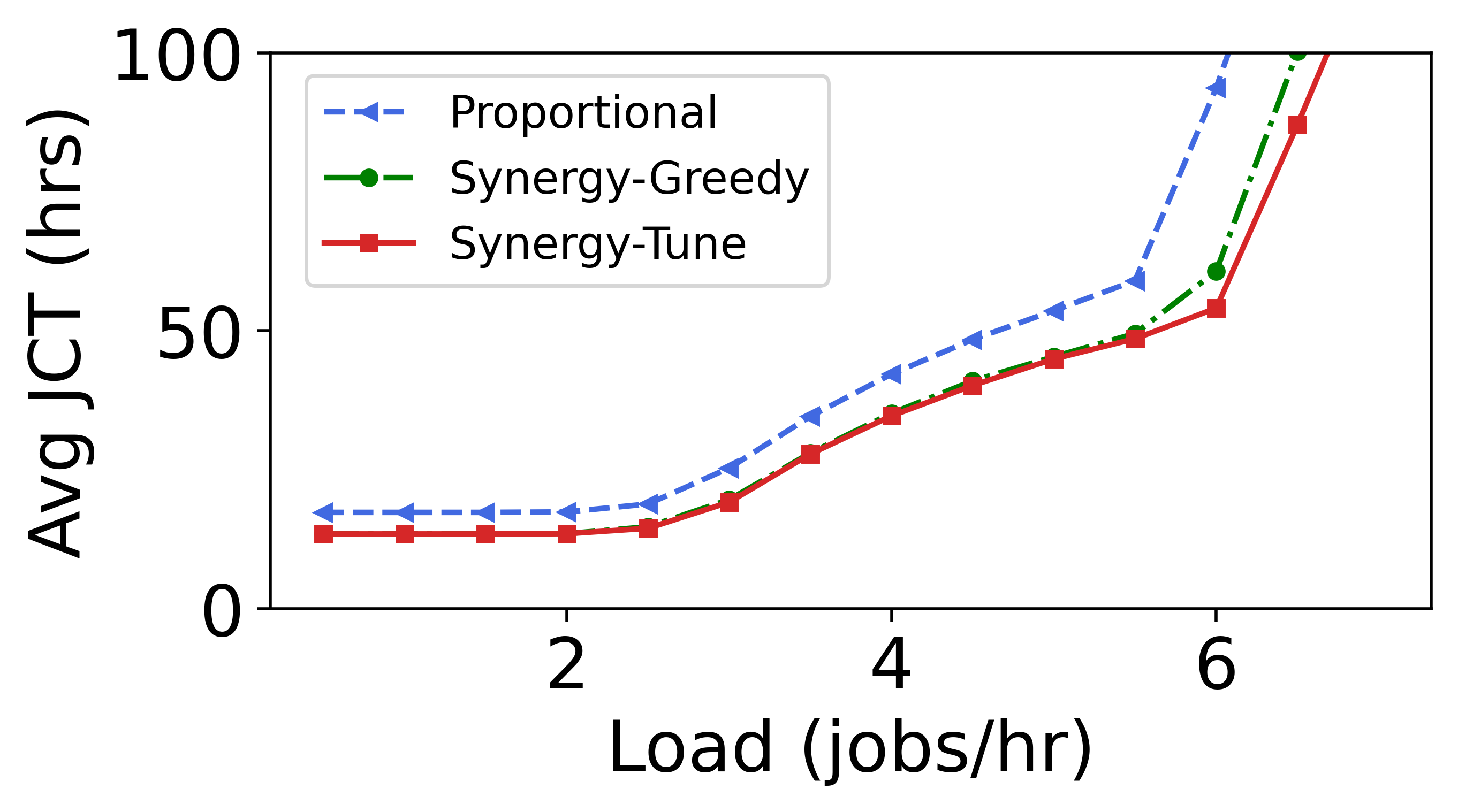} }}
\qquad
\subfloat[Split=(50,0,50)\label{fig-sim-fifo-s3}]{{\includegraphics[width=.3\textwidth]{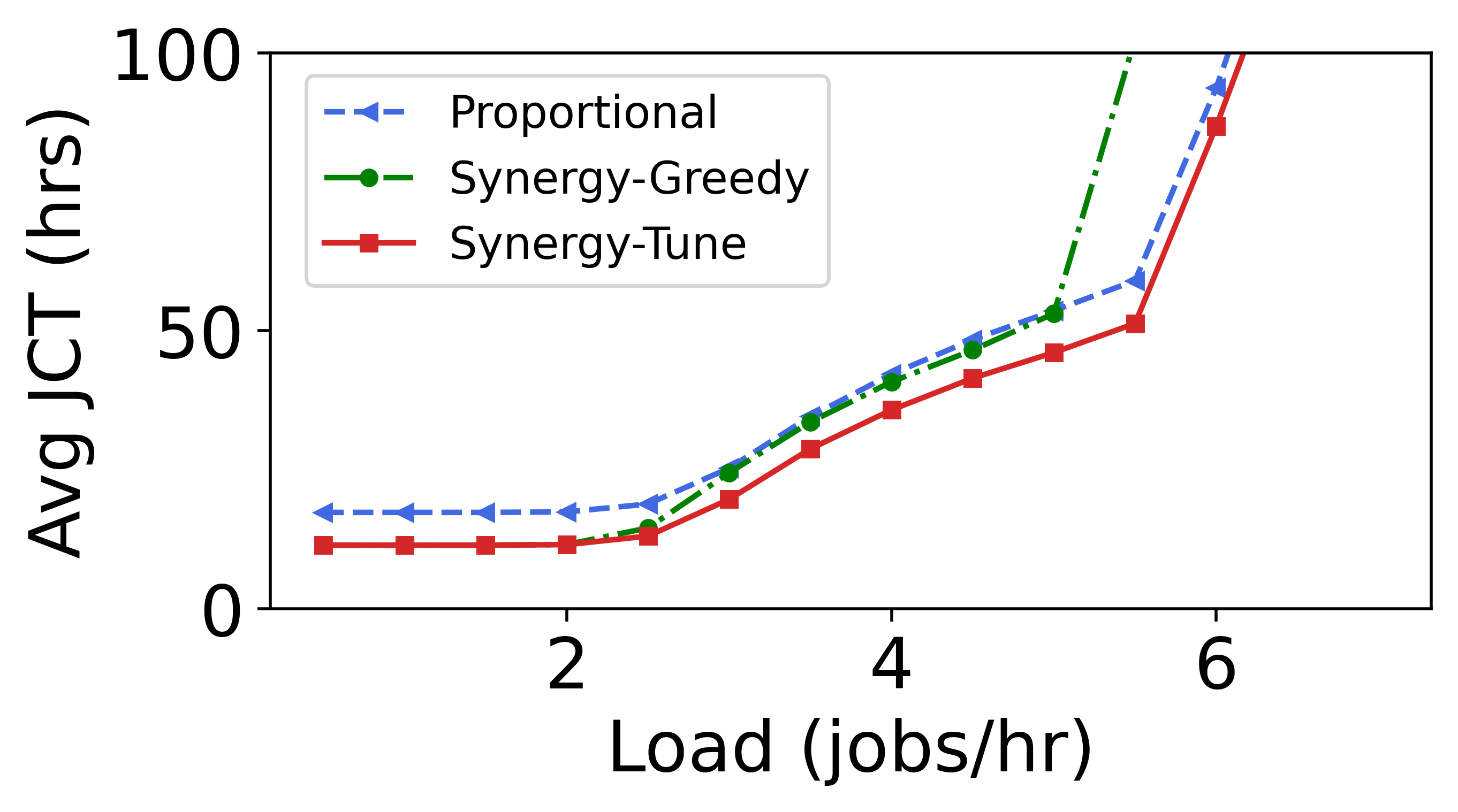} }}%
	
\vspace{-1em}
 \caption{Evaluation of \sysname with varying workload split}
\vspace{-1em}
\label{fig-fifo-vary}
\end{figure*}
  \begin{figure*}[h]
	\centering
 \subfloat[Ratio 4\label{fig-cpu-4}]{{\includegraphics[width=.3\textwidth]{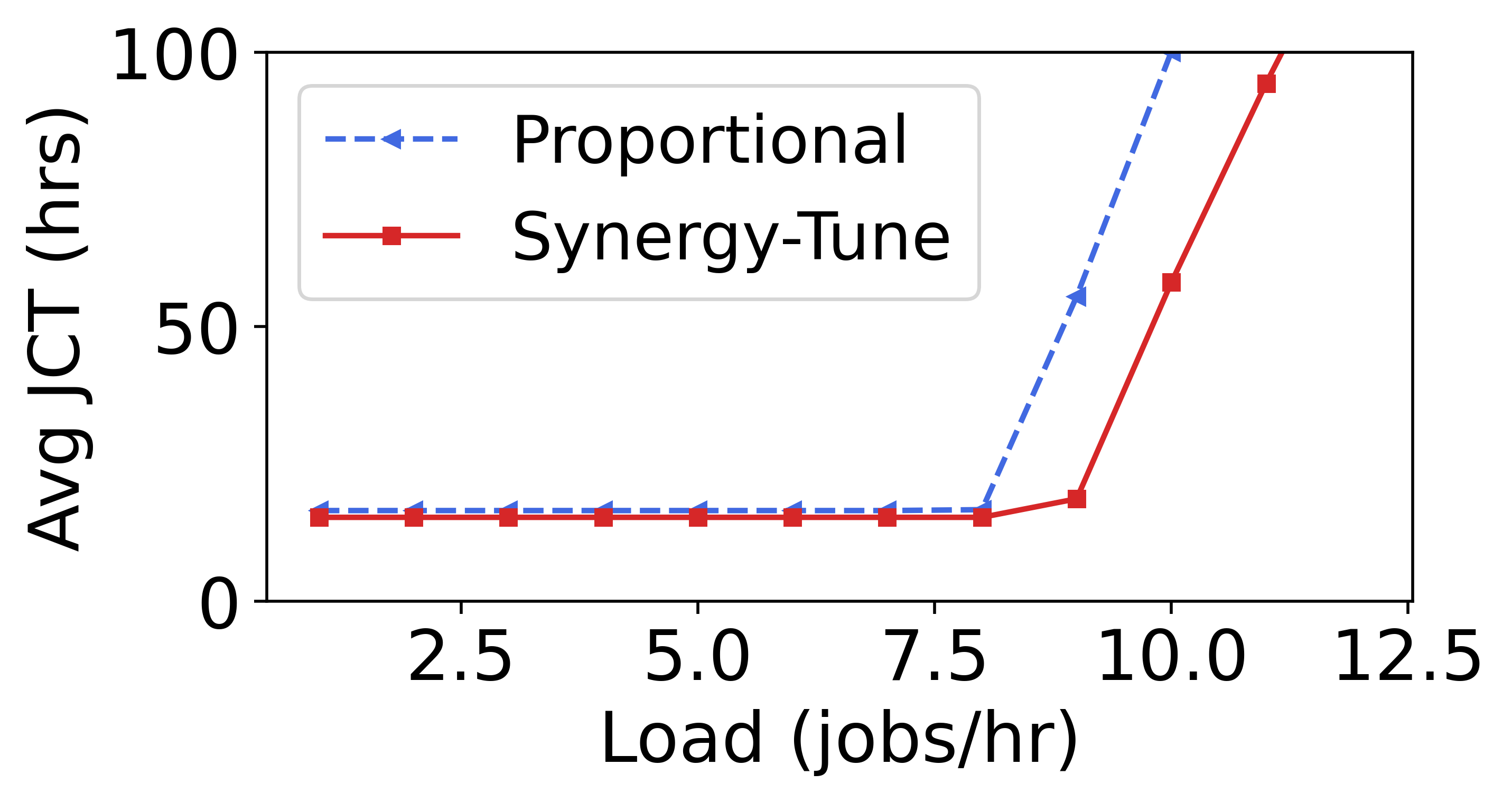} }}%
\quad
\subfloat[Ratio 5\label{fig-cpu-5}]{{\includegraphics[width=.3\textwidth]{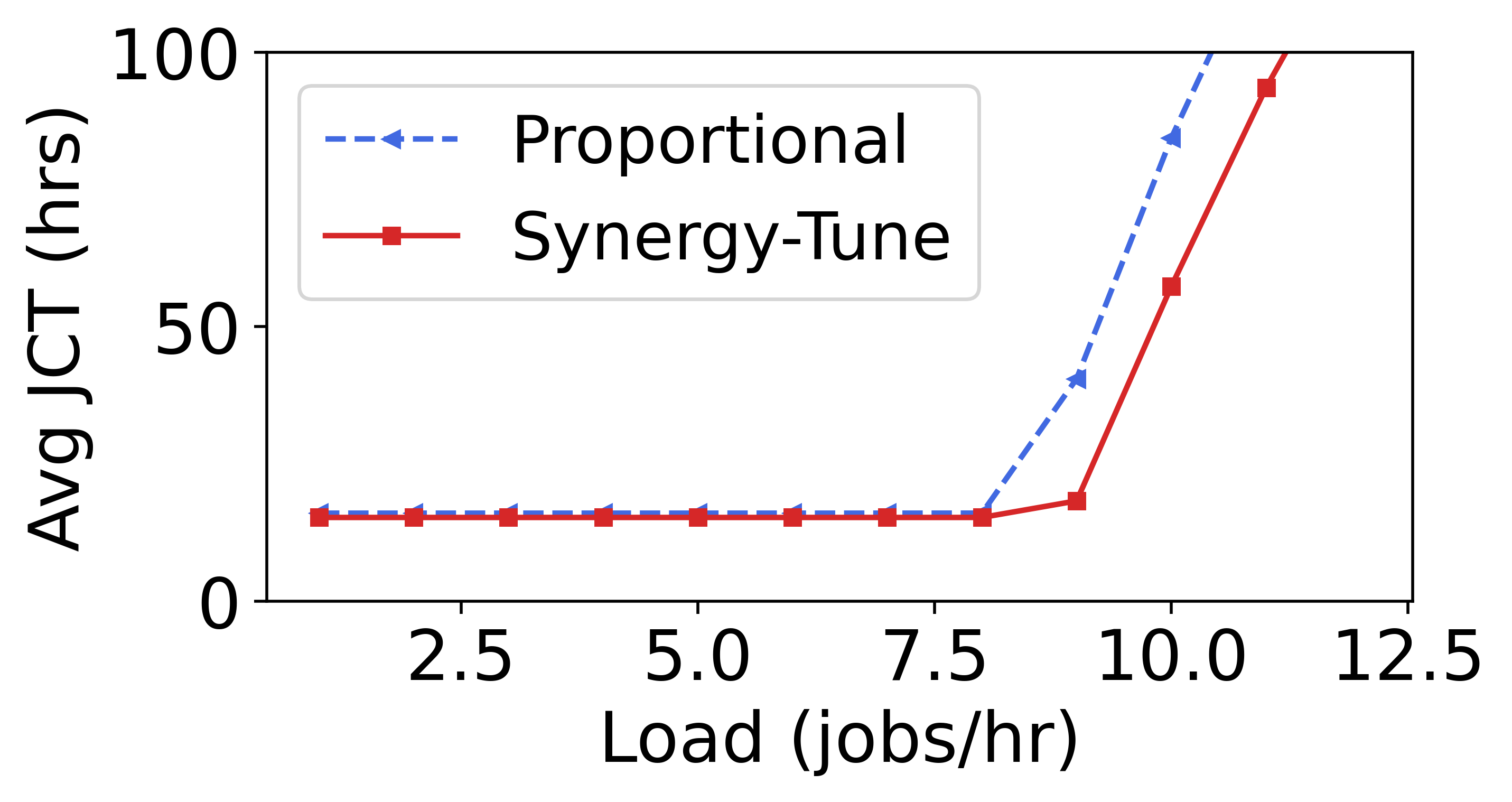} }}
\quad
\subfloat[Ratio 6\label{fig-cpu-6}]{{\includegraphics[width=.3\textwidth]{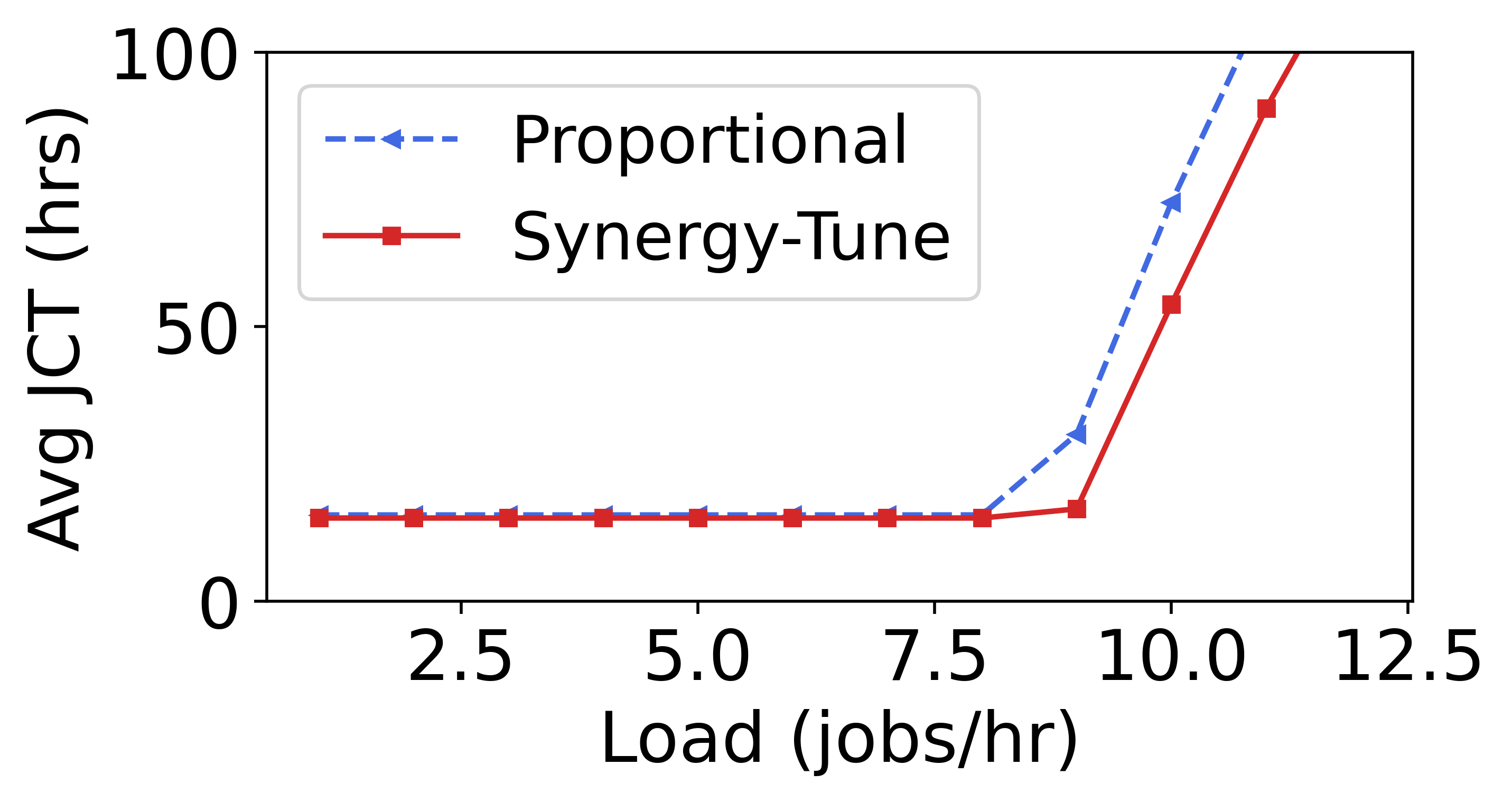} }}%
\vspace{-1em}
 \caption{Evaluation of \sysname across different  CPU:GPU Ratio}

\label{fig-eval-vary-cpu}
\end{figure*}

Similarly, for FTF scheduling policy, \systune observed  2.3\myx and 2\myx improvement in average JCT for a single-GPU and multi-GPU trace respectively.

\subsection{Impact of workload split}
\label{sec-eval-split}

 Workload split decides the percentage of resource sensitive jobs in the workload. As the percentage of speech and image models increase in the trace, there may not be enough spare CPU and memory resources to perform disproportionate allocation, as they are mostly CPU- and memory-hungry. Figure~\ref{fig-fifo-vary} plots the average JCT with varying load for 3 different workload splits with FIFO scheduling for multi-GPU jobs. As the percentage of resource-sensitive jobs increase, we observe that \sysgreedy breaks down, and ends up degrading JCTs significantly compared to a \fair allocation. This is because, the naive greedy technique results in resource fragmentation when the demand along CPU and memory dimensions are high, leaving several GPUs underutilized. Whereas, by the design of \systune, it allocates at least as many resources required to achieve the throughput of \fair allocation; therefore, even in the worst case workload split shown in Figure~\ref{fig-sim-fifo-s3}, where all the jobs are CPU- and memory-sensitive, \systune performs as good as \fair allocation.

 \vheading{Resource utilization}. Figure~\ref{fig-eval-sim-gpufrag} plots the GPU allocation over time for the workload in Figure~\ref{fig-sim-fifo-s3} at a load of 5.5 jobs/hr where the cluster GPU demand is higher than 100\%. While \systune is able to sustain a higher load by finishing jobs faster, \sysgreedy severely under-utilizes GPU resources throughout the workload, trading it off for higher CPU and memory allocation. At low loads, as shown in Figure~\ref{fig-eval-cpu-util}, 
 \fair allocation only utilized 60\% of the available CPU resources,  while \systune utilized it \upto a 90\%, resulting in  1.5\myx lower average JCT.

\subsection{Impact of CPU:GPU ratio}
\label{sec-eval-cpu-cores}
While our prior experiments assume a CPU:GPU ratio of 3 (similar to the NVIDIA DGX-2), Figure~\ref{fig-eval-vary-cpu} plots the average JCT for a FIFO scheduler on a single-GPU trace as we increase cluster load and vary the CPU:GPU ratio from 4 to 6 (corresponding to other server SKUs in Table~\ref{fig-mot-sku}).  As the CPU:GPU ratio in a server increases, the baseline \fair scheduler gets more CPU cores per GPU, thereby reducing data stalls in the baseline. This in turn, reduces the gap between \fair and \systune. Despite that, at a load of 9 jobs/hr, \systune lowers the avg JCT by 3.4\myx, 3\myx, 2.2\myx, and 1.8\myx for a CPU:GPU ratio for 3, 4, 5 and 6 respectively.

\subsection{Comparison to \sysopt}
\label{sec-eval-opt}
Calculating optimal allocations for every scheduling round with \sysopt can be quite expensive, especially for large cluster sizes. We experimentally validated that the time taken for per-round allocations for \sysopt increases exponentially with increasing cluster sizes, while that for \systune is hardly a second. 
We also show experimentally that the allocations given by \systune are close to those estimated by \sysopt in \sref{sec-eval-deploy} and \sref{sec-eval-sim}. For a cluster size of 128 GPUs used in our experiments, \systune converges at allocations that are within 10\% of the optimal value, 200\myx faster than \sysopt.

\begin{figure}[!t]
\vspace{-1em}
\centering
 \includegraphics[width=.48\textwidth]{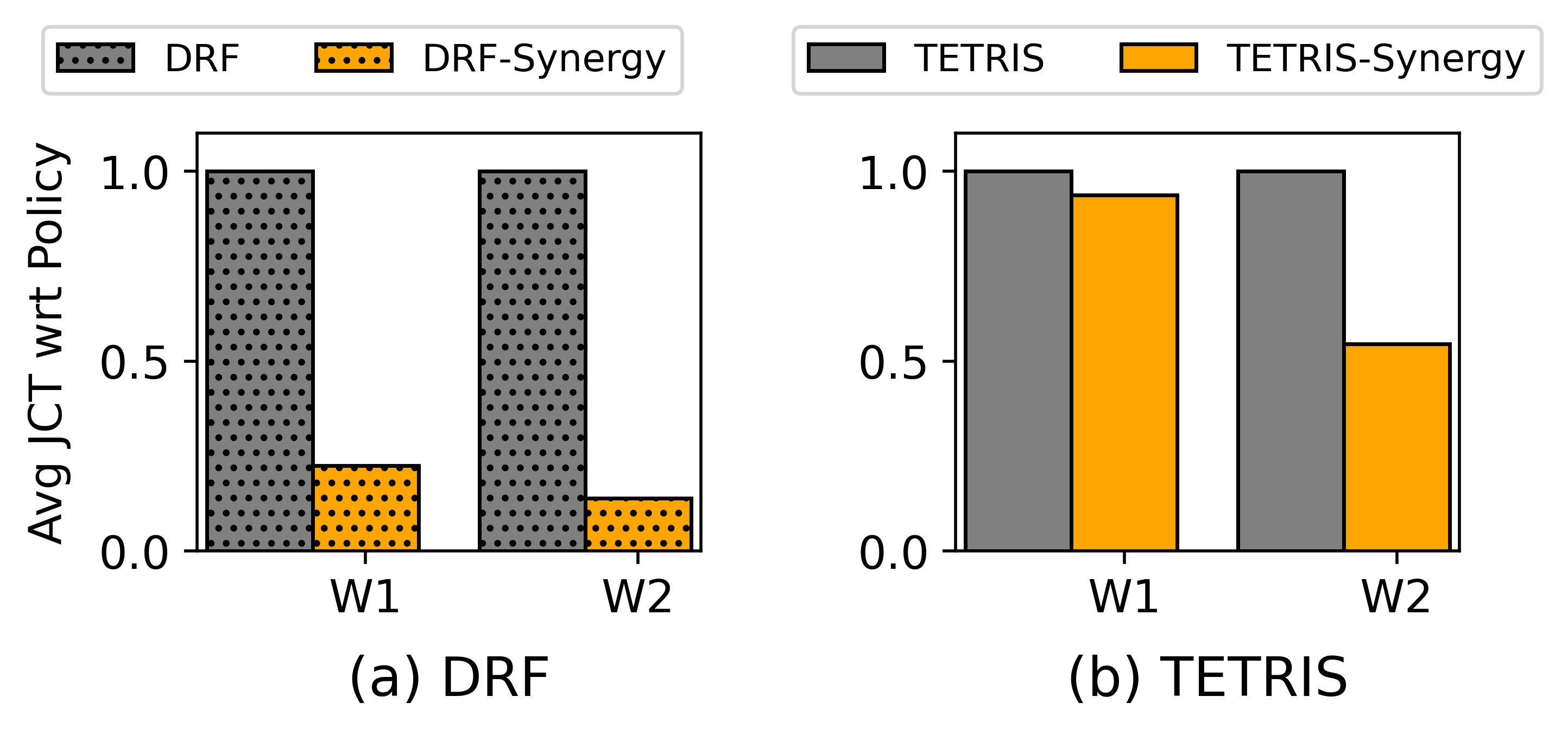}%
 \vspace{-1em}
\caption{Comparison to big data scheduling policies}
\label{fig-eval-bigdata}
\vspace{-1em}
\end{figure}

\subsection{Comparison to DRF and Tetris}
\label{sec-eval-big-data}
Big data schedulers like Dominant Resource Fairness (DRF)~\cite{drf} and Tetris~\cite{tetris} have explored multi-dimensional resource allocation for map-reduce jobs. DNN jobs have different properties when compared to big-data jobs. DNN jobs are gang-scheduled, meaning they can run only when all the GPUs requested by them are available on the cluster at once. Further, the auxiliary resource requirements like CPU and memory are fungible unlike the GPU demand. DRF and Tetris assume resources to be statically allocated throughout the lifetime of a job, whereas \sysname assumes these resources to be fungible and could result in varied allocations throughout the lifetime of a DNN job. Furthermore, profiling the DNN job's resource demands is unique to \sysname; big data schedulers assume that the job request already encodes resource demands across all dimensions. To evaluate \sysname against these policies, we assume that the best-case resource requirement for CPU and memory is fed as input to the bigdata scheduling policies using \sysname's profiling mechanism.

On a cluster of 128 GPUs, we evaluate these policies on two different workload compositions : W1 (20,70,10), and W2 (50,0,50) and compare the naive policy with its \sysname-variant, which allows resource tuning. W1 represents a workload split with a good mix of resource-sensitive as well as resource-insensitive jobs. W2 is a workload dominated by resource-sensitive jobs, which is one of the worst-case scenarios for multi-dimensional scheduling as it could lead to GPU fragmentation (explained in \sref{sec-eval-split})  

We plot the results in Figure~\ref{fig-eval-bigdata}. Tuning resource allocation across jobs using \sysname reduced the average JCT of DRF by 7.2\myx and that of Tetris by 1.8\myx for the  workload split W2. This is because \sysname is able to allocate auxiliary resources in a fungible-manner every round, whereas the big-data scheduler's static allocations performs similar to greedy techniques, resulting in GPU fragmentation, and thereby degrading the overall cluster metrics.
\sysname performs the best in each scenario as it uses the best-case resource demands of jobs to perform fungible, disproportionate allocation. 

	\section{Discussion and Future Work}
\label{sec-disc}
In this section, we elaborate on some of the assumptions made by \sysname, derived from our experiences with large scale deployed cluster schedulers at Microsoft, and discuss 
what happens if these assumptions are relaxed.

\vheading{\revised{Homogeneous clusters}}. \revised{Scheduling in \sysname assumes that the GPU cluster is homogeneous. This assumption is based on the practical observation that our clusters have thousands of accelerators per homogeneous cluster~\cite{philly}. While there is heterogeneity in hardware across clusters, it is often the case that users select one homogeneous cluster to run their job in production.
For instance, a production cluster could have  two homogeneous virtual clusters (VCs), each comprising of a specific generation of GPU. Each VC is managed separately, and assigned to a specific task - training or inference, for predictable performance. 
While recent works have explored the impact of blurring these boundaries and scheduling across heterogeneous hardware ~\cite{gavel, gandivafair, allox}, such co-scheduling poses several practical challenges~\cite{pai}. For example, some tasks such as low-latency inference are business-critical, user-facing applications which need to run on specific hardware, and need data isolation. Others have specific GPU memory requirements, or need advanced hardware features like NVLink. Hence, users in our production settings specify a specific instance type to run each of their jobs on. Hence it is useful for a scheduler to optimize resource utilization in the context of homogeneous clusters. That said, \sysname’s ideas can also be extended to a heterogeneous cluster by profiling CPU and memory requirements along an additional dimension - GPU type, at an additional profiling cost. The optimal algorithm can then maximize throughput based on a 3-dimensional resource-sensitivity matrix $W_j$. We present the formulation for this in the extended version of the paper~\cite{synergy-arxiv}.} 

\vheading{\revised{Use of MinIO}}. \revised{\sysname assumes the use of MinIO~\cite{coordl} because it is a DNN-aware caching mechanism that outperforms traditional OS page caching and allows performance predictability. It provides resource isolation and reduces storage fetch stalls~\cite{coordl}. If we do not use MinIO, we will have to profile the model at discrete memory allocations which will increase the profiling costs, and also potentially change the trends in profiling matrix. } 

\vheading{Preprocessing overhead}. Preprocessing for vision tasks includes random cropping and transformations of the image in the critical path. Reusing the same transformed images across epochs hurts accuracy~\cite{refurbish, coordl, checkfreq}, whereas it is practically infeasible to pre-process offline due to the prohibitive storage cost (dataset size * num epochs).  It is possible to alter the extent of CPU intensiveness by varying the number of augmentations performed. In this work, we have assumed that the augmentations required for each model are as specified by the published models themselves and we do not change this so as to not affect accuracy.  On the horizon, we do observe recent schemes such as RandAugment~\cite{randaugment}, AutoAugment~\cite{autoaugment} which consider more computationally-intensive augmentation schemes (and associated accuracy gains). Such a rising trend in extreme preprocessing, makes a strong case for a system like \sysname.

\vheading{\revised{Sharing storage and network}}. 
\revised{In our paper, we show how to reallocate CPUs and memory across jobs resident on the same server, for example, by co-locating a CPU-intensive task with a non CPU-intensive task. For our DNN training jobs, we assume that a dataset is downloaded locally and loaded into server memory when the job is started (constrained by the memory allocation limits). Prior work has similarly looked at co-locating network-intensive jobs with non network-intensive jobs~\cite{themis, tiresias}, but unlike \sysname, re-allocation of shared network bandwidth is not explicitly handled by those schedulers. We leave it to future work  to explore how ideas in \sysname can also be extended to reason about demands that individual jobs place on storage and network bandwidths.}


\vheading{\revised{GPU elasticity and sharing}}. \revised{While some recent works explore transparently changing the GPU allocation during the life of a job~\cite{pollux}, the impact of changing batch sizes and  hyperparameters on training accuracy is unclear for a wide variety of tasks. It is therefore practical to assume that the GPU demand of a job is constant throughout its lifetime as is the case for jobs in our production clusters.}

\revised{\sysname works by improving the throughput of jobs that are bottlenecked on data stalls. For jobs that have data stalls, GPU efficiency cannot be improved by multiplexing (spatial sharing) because they are waiting for input data. However, for a subset of jobs that are insensitive to auxiliary resource allocation, GPUs could be multiplexed between jobs. It would be interesting to explore how to impart resource-sensitivity awareness alongside GPU spatial sharing, which we leave for future work.}

\vheading{\revised{Tradeoff between consolidation and allocation}}. \revised{When multi-GPU jobs are split across physical servers, they may incur a penalty due to network communication~\cite{narayanan2019pipedream, gandiva}.  DNN jobs therefore prefer consolidation. In this work, we assume that no more than a server's worth of CPU or memory resources can be allocated to a job if its GPU demands can be satisfied by one server. However, we find that some jobs may benefit from giving up consolidation if the throughput gain due to increased CPU or memory allocation is higher than the penalty due to splitting. We leave the exploration of  the trade off between consolidation and allocation, while taking into account the network overhead, to future work.}

\vheading{\revised{Leveraging model and pipeline parallelism}}. \revised{Our evaluation assumes distributed data-parallel jobs. But model and pipeline parallel execution schemes also have an input stage that ingest and pre-process data.  Unlike data-parallel training, each stage in the pipeline might have a different CPU-GPU and memory-GPU requirement.  While these jobs would have to be profiled to identify the CPU and memory sensitivity of each stage of the pipeline, \sysname’s contributions directly carry forward to such settings.}

\section{Related Work}
\vheading{DNN cluster schedulers}.
A number of recent schedulers for DNN workloads each focus on improving a certain objective; Cluster utilization (Gandiva~\cite{gandiva}), JCT (Tiresias~\cite{tiresias}), and fairness (Themis~\cite{themis}, Gandiva-Fair~\cite{gandivafair}). Some have also looked at exploiting performance heterogeneity among accelerators to improve cluster objectives~\cite{gavel, allox}. All these schedulers assume GPU to be the dominant resource in the scheduling task; i.e., a user requests a fixed number of GPUs for her DNN job, and when the requested number of GPUs are all available, the job is scheduled to run. Rather than allocating a fixed number of GPUs, building on GPU-elasticity for a single job~\cite{andrew}, some recent schedulers like AFS~\cite{afs} and Pollux~\cite{pollux} leverage throughput metrics to provide GPU elasticity in multi-tenant clusters (in addition to tuning batch size and learning rate).  However, in all these cases, auxiliary resources such as CPU and memory are allocated proportional to the number of GPUs allocated to the job. Existing schedulers thus ignore \textit{resource-sensitivity} of the DNN tasks to CPU and memory.  \sysname shows that, irrespective of the number of GPUs allocated, auxiliary resource-sensitive allocation is crucial to achieve better cluster utilization.

\vheading{Big data schedulers}. Our work builds upon the insights drawn from the rich literature of schedulers for big data jobs~\cite{yarn, mesos, graphene, carbyne, tetris, drf}. Big data schedulers like Tetris~\cite{tetris}, and DRF~\cite{drf} have looked at the problem of multi dimensional resource allocation for big data jobs. They propose new scheduling policies aimed at optimizing a specific cluster objective for jobs whose resource demands are prior known. 
In contrast, the primary resource in a DNN job is the accelerator (GPU), whose requirement is specified by the job; other resources are fungible. Our work exploits this insight to perform disproportionate allocations by profiling job resource sensitivity, and then appropriately packing them onto servers.  

\vheading{Data stalls}. Recent, deep characterization studies explored the impact of CPU and memory on individual DNN jobs~\cite{coordl, tfdata} 
Unlike prior work that focuses on individual jobs, the focus of our paper is on the tricks we can play when we schedule multiple jobs together in a cluster. 

\vheading{\revised{Disaggregated data prep}}. \revised{There have been recent orthogonal efforts that aim at reducing the cost of data preprocessing, and thereby the load on CPUs using disaggregated data prep~\cite{fbstalls}. However, one has to pay the network cost of shuffling preprocessed tensors, which could quickly become the bottleneck especially for vision models with rich datasets. \sysname on the other hand, assumes standard pre-processing pipelines at the training servers, and aims to reduce the cost of pre-processing using better resource allocation.} 
	
\section{Conclusion}
This paper introduces \sysname, a resource-sensitive scheduler for DNN training jobs. \sysname is based on the insight that not all jobs exhibit the same level of sensitivity to CPU and memory allocation during DNN training; breaking the shackles of \fair allocation can result in improved utilization of existing cluster resources and improved job and cluster-wide objectives. Our experiments on physical and large simulated clusters show that \sysname can reduce average JCT by upto 3.4\myx over \fair allocation. 

	\bibliographystyle{plain}
	\bibliography{all}
	\appendix

\section{Appendix}

\subsection{\sysopt}
In this section, we describe the formulation and proof for \sysopt in a homogeneous GPU cluster.

\label{appendix-sec-opt}
\vheading{Problem Definition}.
Our goal is to allocate CPU and memory to each job so as to maximize the throughput, while guaranteeing that each job makes at least as much progress as it would do if we allocate its {\em \fair share}. 

\vheading{Notation}
\begin{itemize}
\item $s$: The number of machines or servers.
\item For each machine $i \in [s]$, we denote $G_i, C_i, M_i$ as the total GPU, CPU, and memory available on machine $i$.
\item We denote the total GPU available across all machines by $G$. That is, $G = \sum_{i} G_i$. Similarly, we denote $C, M$ as the total CPU and Memory capacity across all machines.
\item We denote jobs by indices $j$. The GPU requirement of job $j$ is denoted by $g_j$.
\item For each machine $i \in [s]$, we denote $C_{g}, M_{g}$ as the \fair allocation of CPU and memory. That is, $C_{g} = C_i/G_i * g_j$ and $M_{g} = M_i/G_i * g_j$.
\item $J_t$: The set of jobs that needs to be scheduled in each {\em round}.  $J_t$ is the set of runnable jobs for this round, identified by the scheduling policy such that the total GPU requirements of jobs in $J_t$ is at most the total GPU capacity of the cluster. In notation, $\sum_{j \in J_t} g_j \leq G$. 
\item $n$: We denote the number of jobs in the set $J_t$ by $n$. In notation, $n = |J_t|$.
\item $W_j$: We assume that resource sensitivity matrix for each job $j$ is given as input. $W_j[c, m]$ denote the amount of progress made on job $j$ per round if $c$ units of CPU and $m$ units of (RAM) memory are allocated to job $j$.
\item With a baseline \fair allocation strategy the progress a job makes in each round is equal to $W[C{g}, M{g}]$.
\end{itemize}

\subsubsection{Throughput Upperbound in an Optimal Solution}
It is not hard to show that our problem is NP-hard.
So, we resort to finding approximate solutions.
Towards that we first find an {\em upperbound} on the throughput achievable by an optimal solution.
We achieve that by formulating our problem as a linear program (LP).
Moreover, we assume an {\em idealized setting}: We assume that all the CPU and memory available across all the machines is present in one (super) machine.
That is, there is only one machine with $C$ units of CPU and $M$ units of memory.
Note that in reality $C$ units of CPU and $M$ units of memory are spread across $s$ machines.
This means that in our throughput allocation, we do not take into account the effect of network when resources are allocated to jobs across multiple machines.
Therefore, the true optimal solution of our problem can only do worse than the idealized allocation.

\subsubsection{An LP formulation}
We get an upperbound on the optimal allocation via an LP formulation.  
The variables of our LP are denoted by $y_{\{c,m,j\}}$, which should be interpreted as follows.
If for a job $j \in J_t$,  $y_{\{c,m,j\}} = 1$, then it means that in the LP solution $c$ units of CPU and $m$ units of memory are allocated.
We further note that for every job $j$, there is a variable $y_{\{c,m,j\}}$  for {\em for every possible} allocation of CPU and memory. We consider these variables in the discrete space as identified by our resource sensitivity matrix. 


\begin{itemize}
\item Our objective function is to maximize the throughput. We formulate it as follows using our LP variables.

\begin{equation}
\text{Maximize}  \quad \sum_{j \in J_t}  \sum_{[c,m]} W_j[c, m] \cdot y_{\{c,m,j\}}
\end{equation}

Now, we enforce constraints such that LP solution is feasible in the idealized setting we talked about.

\item First constraint we enforce is that the total CPU allocated to jobs is no more than the total capacity available:

\begin{equation}
 \sum_{j \in J_t}  \sum_{[c,m]} c \cdot y_{\{c,m,j\}}  \leq C
\end{equation}

\item Similarly, we make sure that the total memory allocated to jobs is no more than the total capacity available:

\begin{equation}
 \sum_{j \in J_t}  \sum_{[c,m]} m \cdot y_{\{c,m,j\}}  \leq M
\end{equation}

\item We want LP to allocate only one configuration of CPU and memory to each job.
\begin{equation}
\text{ For each job $j \in J_t$:} \quad  \sum_{[c,m]} y_{\{c,m,j\}}  = 1
\end{equation}

\item Finally, we want LP solution to be as good as the fair allocation.

\begin{equation}
\text{ For each job $j \in J_t$:} \quad  \sum_{[c,m]} W_j[c,m] \cdot y_{\{c,m,j\}}  \geq W_j[C_{g}, M_{g}] 
\end{equation}
\end{itemize}

\begin{theorem}
The throughput achieved by the LP (1-5) is at least the throughput achieved by an optimal solution to our problem.
\end{theorem}

\begin{proof}
Consider an optimal solution $O$ to our problem. Suppose job $j$ receives $c^*$ units of CPU and $m^*$ units of memory in $O$.
Then we define the following feasible solution to our LP (1-5): Set $y_{c^*, m^*, j} = 1$.
Clearly, this is a valid solution and satisfies constraints (1-4).
\end{proof}

It is easy to verify that the optimal solution for our problem defines a feasible solution to our LP.
On the other hand, as the LP solution can be fractional, in the sense $y_{\{c, m, j\}}$ variables can take fractional values, the optimum solution for LP can be no smaller than the true optimum solution, and thus always an upperboun on the throughput one can achieve for our problem.
By enforcing the integrality constraints on $y_{\{c, m, j\}}$ variables one can getting a tighter upper bound.
Indeed, in our experiments we solve this as a Integer Linear Program (ILP) where $y_{\{c,m,j\}}$ takes boolean values. For every job, we define the total CPU ($c^*_j$) and memory ($m^*_j$) allocated by the optimal ILP solution as follows.
	\begin{equation}
	\text{For each job $j$, define}   \quad c^*_j := c  \quad \text{if}  \quad y_{\{c,m,j\} == 1}. 
	\end{equation}
	
		\begin{equation}
	and   \quad m^*_j := m  \quad \text{if}  \quad y_{\{c,m,j\} == 1}. 
	\end{equation}

\subsubsection{Feasible Allocation on Multiple Machines}
\label{appendix-second-lp}

Recall that in the LP(1-5), we assumed that all the resources are present on a single machine.
However, in reality these resources are spread across multiple machines.
So, now we need to make an allocation taking into account this fact.
We achieve that by solving another linear program.

Now our goal is the following:

\begin{itemize}
\item For each job $j \in J_t$, allocate $g_j$ units GPU, $c^*_j$ units of CPU, and $m^*_j$ units of memory across $s$ machines such that each job is fully scheduled on a single machine. We call $(g_j, c^*_j, m^*_j)$ as the demand vector of job $j$.
\end{itemize}

Again, the above problem is an instance of multi-dimensional bin packing problem, so it is NP-hard. 
Instead, we try to reduce the number of jobs that get fragmented.
So, our new goal is:
\begin{itemize}

\item For each job $j \in J_t$, allocate $g_j$ units GPU, $c^*_j$ units of CPU, and $m^*_j$ units of memory across $s$ machines such that the number of jobs that get fragmented is minimized.

\end{itemize}

\vheading{A Feasible Allocation via Second LP}.
\label{sec:allocationLP}
The variables of the second LP are denoted by $x_{i,j}$.
Here index $i$ denotes the machine and $j$ denotes the job.
The variables $x_{i,j}$ are interpreted as follows:
if $x_{i,j} = 1$, it means that resources of job $j$ (that $g_j$ units of GPU, $c^*_j$ units of CPU, and $m^*_j$ units of memory) are allocated on machine $i$.

Now we are ready to find a feasible allocation minimizing the number of fragmented jobs.

\begin{itemize}
\item First constraint we enforce is that the total GPU allocated to jobs is no more than the total capacity available on the  machine:

\begin{equation}
 \text{For each machine $i$ in $[s]$:}  \quad \sum_{j \in J_t} g_j \cdot x_{i,j}  \leq G_i
\end{equation}

\item Next, we make sure that the total CPU allocated is no more than the total capacity available on the machine:

\begin{equation}
 \text{For each machine $i$ in $[s]$:}  \quad  \sum_{j \in J_t}  \sum_{[c,m]} c^*_j \cdot x_{i,j}  \leq C_i
\end{equation}

\item Similarly, we make sure that the total memory allocated is no more than the total capacity available on the machine:

\begin{equation}
 \text{For each machine $j$ in $[s]$:}  \quad  \sum_{j \in J_t}  \sum_{[c,m]} m^*_j \cdot x_{i,j}  \leq M_i
\end{equation}

\item We make sure that every job is allocated all the resources it demands:

\begin{equation}
 \text{For each job $j \in J_t$}  \quad  \sum_{i \in [s]}  x_{i,j} \geq 1
\end{equation}

\item Finally, We make sure that variables are positive.

\begin{equation}
 \text{For each job $j \in J_t$ and $i \in [s]$}  \quad x_{i,j} > 0
\end{equation}
\end{itemize}

Using linear programming theory, we now prove a structural property about our LP that states that most of the variables are integral.
\begin{theorem}
Suppose we assume that no job demands more CPU, GPU or memory available on a single machine.
Then, the total number of jobs that get fragmented in the LP (8-12) is at most 3s.
\end{theorem}

\begin{proof}
Let $\{x^*_{i,j}\}_{i,j}$ denote an optimal solution to the LP(8-12). 
We know from linear programming theory that for every LP, there is an optimal solution which is a {\em vertex} of the polytope~\cite{lp1,lp2,lp3}.
Let $P$ denote the set of positive variables in the LP solution. That is set of $x_{i,j}$ such that  $x_{i,j} > 0$. 
A vertex solution is defined by a linearly independent family of tight constraints.
A tight constraint means that in the LP solution a constraint is satisfied with an equality ($=$).
A tight constraint of the form $x_{i,j} = 0$, only leads to variables not in P.
Therefore, we only we need to consider tight constraints of the form (8), (9), (10), and (11).
Therefore, number of variables taking positive values in $P$ is bounded by 

\begin{equation}
|P|  \leq 3s + n
\end{equation}

The above equation is true because there is only 1 constraint of the type (8), (9), and (10) for each machine and there are $s$ machines.
Further more there is one constraint of type (10) and there are $n$ jobs.

Now let $N_1$ denote the number of jobs that got fragmented in the LP (8-12) solution.
Now each such job contributes at least 2 variables to $P$.
This implies,

\begin{equation}
2N_1 + (n - N_1) \leq |P| \leq 3s + n
\end{equation}

Therefore, $N_1 \leq 3s$, and it concludes the proof.

\end{proof}

\subsection{\revised{Extending to Heterogeneous GPU Clusters}}
\label{sec-appendix-heterogeneous}

\revised{Assume a heterogeneous GPU cluster consisting of several \textit{types} (or generations) of GPU machines.}

\subsubsection{Notation}
\begin{itemize}
\item $K$ : The set of different \textit{types} of machines, where each entry represents a homogeneous group of servers of the same SKU (GPU generation).
\item $s_i$: The number of machines of each type $i$ $\in$ $K$.
\item For each type of machine $i$, we denote $G_i, C_i, M_i$ as the GPU, CPU, and memory available on each machine of type  $i$.
\item We denote the total GPU available across all machines by $G$. That is, $G = \sum_{i} G_i * s_i$. Similarly, we denote $C, M$ as the total CPU and memory capacity across all machines.
\item We denote jobs by indices $j$. The GPU requirement of job $j$ is denoted by $g_j$.
\item $J_t$: The set of jobs that needs to be scheduled in each {\em round}.  We assume that $J_t$ is given to us as input. Moreover, the total GPU requirements of jobs in $J_t$ is at most the total GPU capacity of the cluster. In notation, $\sum_{j \in J_t} g_j \leq G$. 
\item $n$: We denote the number of jobs in the set $J_t$ by $n$. In notation, $n = |J_t|$.
\item (Heterogeneous Case) $W_{ij}$: We assume that throughput matrix for each job $j$ is given as input. $W_{ij}[c, r]$ denote the amount of progress made on job $j$ per round if $c$ units of CPU and $m$ units of (RAM) memory are allocated to job $j$ on machine type $i$.
\end{itemize}

\subsubsection{Problem Definition}
\revised{The goal is to allocate CPU and memory to each job so as to maximize the throughput, while guaranteeing that each job makes at least as much progress as it would do if we allocate its {\em fair share}. 
A naive fair allocation would evenly distribute $1/n$-fraction of each resource to every job. 
Therefore, progress a job makes in each round is equal to $W[C/n, M/n]$ if the cluster is homogeneous.
In the heterogeneous case, we assume that throughput achieved using an appropriate definition fair allocation is given to us by an oracle, for each round. This could be using a heterogeneity and fairness aware scheduler in literature~\cite{gavel, gandivafair}.
We denote this throughput by $W^{\text{Fair}}_j$. An additional constraint we impose in the heterogeneous setting is that, in a given round, a job cannot be split across two different \textit{types} of machines, due to the operational challenges involved in running a job across different GPU generations simultaneously. However, we do allow a job to run on different types of machines across rounds. Prior work on heterogeneity aware schedulers~\cite{gavel} make similar assumptions.}

\subsubsection{An LP formulation}
\revised{We get an upperbound on the optimal allocation via an LP formulation.  Similar to the homogeneous case, we assume an {\em idealized setting for the first LP}: We assume that all the CPU and memory available across all the machines of the same \textit{type} is present in one (super) machine. That is, there is one representative machine $i$ for every homogeneous cluster. The variables of our LP are denoted by $y_{\{c,m,i,j\}}$, which should be interpreted as follows. If for a job $j \in J_t$ and a super-machine of type $i$,  $y_{\{c,m,i,j\}} = 1$, then it means that in the LP solution $c$ units of CPU and $m$ units of memory are allocated on super-machine $i$. }


\begin{itemize}
\item Our objective function is to maximize the throughput. We formulate it as follows using our LP variables.

\begin{equation}
\text{Maximize}  \quad \sum_{j \in J_t}  \sum_{[c,m]} W_j[c, m, i, j] \cdot y_{\{c,m,i, j\}}
\end{equation}

Now, we enforce constraints such that LP solution is feasible in the fractional sense.

\item First constraint we enforce is that the total CPU allocated to jobs is no more than the total capacity available on each machine type $i$.

\begin{equation}
 \sum_{j \in J_t}  \sum_{[c,m]} c \cdot y_{\{c,m,i, j\}}  \leq C_i*s_i
\end{equation}

\item Similarly, we make sure that the total memory allocated to jobs is no more than the total capacity available:

\begin{equation}
 \sum_{j \in J_t}  \sum_{[c,m]} m \cdot y_{\{c,m,i, j\}}  \leq M_i*s_i
\end{equation}

\item We want LP to allocate only one configuration of CPU and memory to each job.
\begin{equation}
\text{ For each job $j \in J_t$:} \quad  \sum_{[c,m]} \left(\sum_{i}y_{\{c,m,i,j\}} \right)  = 1
\end{equation}

\item Finally, we want LP solution to be as good as naive fair allocation.

\begin{equation}
\text{ For each job $j \in J_t$:} \quad  \sum_{[c,m]}  \sum_{i} \left(W_j[c,m,i] \cdot y_{\{c,m,i,j\}} \right)  \geq  W^{\text{Fair}}_j
\end{equation}
\end{itemize} 

\revised{We solve this optimization problem as a Integer Linear Program (ILP) where $y_{\{c,m,i,j\}}$ takes boolean values, i.e., the job $j$ can be placed on only one super-machine of type $i$. }

\vheading{Improving utilization}. \revised{The above ILP may not be able to find solutions for all the jobs $j$ in the set $J_t$ without splitting them across machines of different types. This leaves spare GPUs in the cluster, which potentially other jobs in the wait queue could be assigned to, without violating our constraints. Therefore, we repeat the optimization problem above for the set of unassigned GPUs from the prior steps, and the next set of jobs $J_t'$ in the wait queue. We can repeat this process until either there are no more jobs in the wait queue, or there are no more unassigned GPUs in the cluster.}

\vheading{Feasible allocation on Multiple Machines}. \revised{The first LP described above identifies the resource allocation for a job on a super-machine of type $i$. While in reality, the resources in this machine are physically split across $s_i$ servers. This problem can be solved exactly as in the homogeneous case using a LP as described in \sref{appendix-second-lp}.}
\end{document}